\documentclass[superscriptaddress,aps,pra, twocolumn, showpacs, nofootinbib, longbibliography, showkeys]{revtex4-2}
\usepackage{amsmath,amssymb,amsthm}
\usepackage{easybmat}
\usepackage[colorlinks=true,citecolor=blue,urlcolor=blue, linkcolor = magenta]{hyperref}
\usepackage[pdftex]{graphicx}
\usepackage{times,txfonts}
\usepackage{braket}
\usepackage{color}
\usepackage{natbib}
\usepackage{comment}
\usepackage{subcaption}
\newcommand{\be}{\begin{equation}}
	\newcommand{\ee}{\end{equation}}
\newcommand{\ba}{\begin{eqnarray}}
	\newcommand{\ea}{\end{eqnarray}}

\newcommand{\tr}{\operatorname{Tr}}

\newtheorem{definition}{Definition}
\newtheorem{thm}{Theorem}
\newtheorem{proposition}{Proposition}
\newtheorem{Lemma}{Lemma}
\newtheorem{cor}{Corollary}

\begin{document}

	\title{Exploring Non-Markovianity in Ergodic Channels: Measuring Memory Retention through Ergotropy}

\author{Ritam Basu\textsuperscript{}}
\email{ritam.basu@research.iiit.ac.in}
    \affiliation{Center for Quantum Science and Technology and Center for Security Theory and Algorithmic Research, International Institute of Information Technology, Gachibowli, Hyderabad, India\textsuperscript{}}
    \author{Anish Chakraborty\textsuperscript{}}
\email{anishsphs51@gmail.com }
    \affiliation{Indian Institute of Science Education And Research Kolkata
Mohanpur, Nadia - 741 246
West Bengal, India}
\author{Himanshu Badhani\textsuperscript{}}
\email{himanshubadhani@gmail.com }
\affiliation{Optics \& Quantum Information Group, The Institute of Mathematical Sciences, C. I. T. Campus, Taramani, Chennai - 600113, India}
    \affiliation{Homi Bhabha National Institute, Training School Complex, Anushakti Nagar, Mumbai 400094, India}
\author{Mir Alimuddin\textsuperscript{}}
\email{aliphy80@gmail.com }
    \affiliation{ICFO-Institut de Ciencies Fotoniques, The Barcelona Institute of Science and Technology, Av. Carl Friedrich Gauss 3, 08860 Castelldefels (Barcelona), Spain.}
       \author{Samyadeb Bhattacharya} 
         \email{samyadeb.b@iiit.ac.in}
        \affiliation{Center for Quantum Science and Technology and Center for Security Theory and Algorithmic Research, International Institute of Information Technology, Gachibowli, Hyderabad, India\textsuperscript{}}

        \begin{abstract}
            In this work, we introduce and characterize a broad class of quantum operations with a unique fixed point, termed quantum ergodic channels. We derive Lindblad-type master equations for these channels in arbitrary finite dimensions and analyze their non-Markovian dynamics using established measures. When the fixed point is a passive state, the channels exhibit ergotropy dynamics with notable thermodynamic implications. Specifically, under Markovian processes, ergotropy—a measure of the extractable work from a system under unitary evolution—monotonically decreases. However, in non-Markovian dynamics, ergotropy fluctuates, leading to a backflow effect that highlights memory-induced resource recovery. Our findings suggest that this ergotropy backflow could serve as an operationally meaningful indicator of non-Markovianity, offering new perspectives on the interplay between memory effects and thermodynamic behavior in open quantum systems. This study enhances the theoretical framework for understanding energy dynamics under ergodic channels and highlights new avenues for exploring the implications of memory effects in quantum batteries.

        \end{abstract}
	
\maketitle

\section{Introduction}
In the past two decades, much efforts have been dedicated to understand and analyse the aspect of quantum non-Markovianity \citep{nm1,nm2,nm3,nm4,nm5,nm6,nm7,nm8,nm9,nm10,nm11,nm12,nm13,nm14,nm15,nm16,nm17,nm19,nm20}. Through the characterization of  complete positive (CP)-divisibility of a channel \citep{nm8} and information backflow \citep{nm6}, a particular characterization of quantum non-Markovianity has been established, where a quantum channel, which can not be realised as a concatenation of arbitrary number of linear CP maps is understood as a non-Markovian quantum channel. Additionally, it has also been established that, for information backflow, i.e. non monotonic behaviour of distance measures which are monotones under CP divisible maps like trace distance (or pseudo distance distinguishibility measures like relative entropies) to occur, it is necessary (but not sufficient) to have CP-indivisible quantum channels \citep{nm21}. Though CP-indivisibility is not the most general identification of non-Markovianity and there are different spectacles to view this interesting dynamical feature \citep{nm20,nm20a}, it is sufficiently general to consider quantum non-Markovianity from the backdrop of indivisible operations. In the course of the ongoing plethora of research, it has been established that quantum non-Markovianity can be considered as a resource in various information theoretic protocols \citep{nm22,nm23,nm24}. Moreover, from the context of experimentally implementable situations, it is natural to consider non-Markovian quantum operations, since for those scenario less stringent apriori physical constraints are involved. It is therefore important to classify non-Markovian quantum dynamics from the perspective of channels with definitive geometric structure. Such analysis enables a robust classification of a scientific phenomenon under ongoing scrutiny to achieve a more profound understanding. Such attempts has been previously made for quantum non-Markovian dynamics for generalized depolarizing channels in arbitrary dimensions \citep{nm25}. Depolarizing channels can be understood as a specific case of unital operations having maximally mixed state as the singular fixed point of the dynamics. Though the analysis of non-Markovianity for generalized depolarizing operations serves its purpose of both qualitative and quantitative characterization of the said phenomenon, it is constrained to a restricted class of quantum operations. It is therefore imperative to generalize the classification for a much more extensive class of quantum channels. In this article, we generalize the classification of non-Markovianity from the backdrop of quantum channels having singular fixed point, intimately connected to an intriguing class of operations coined as ergodic channels.  

A physical process is said to be ergodic if the statistical properties of it can be estimated completely from the long time realisation of the said operation. In the study of thermal relaxation, this concept of ergodicity plays a pivotal role \citep{ergo1,ergo2,ergo2a,extra1,extra2}. It is also found to be important in the theory of quantum control \citep{ergo3,ergo4,ergo5,ergo6}, quantum communication \citep{ergo7} and many more \citep{ergo8,ergo9,ergo10}. We thus classify ergodic quantum channels as quantum channels having unique fixed points. 
In this work, we restrict ourselves to invertible quantum operations; i.e. if $\Phi$ represents a certain quantum operation, taking an element from a Hilbert space $\mathcal{H}_A$ to $\mathcal{H}_B$, then its inverse operation $\Phi^{-1}$ exists, which maps an element of $\mathcal{H}_B$ to an element of $\mathcal{H}_A$. It has been shown previously that for such kind of dynamical maps, a Lindblad type exact master equation always exists \citep{master1,master2}. We construct generalized Lindblad like master equations for such kind of invertible channels in arbitrary dimensions, which are ergodic. It is notable that the generalized depolarizing channels fall into a subcategory of the more general class of operations that we consider here. In other words, we generalize the study presented in \citep{nm25}, which is restricted to depolarizing channels in arbitrary dimensions.  

Intriguingly, we observe that the class of ergodic channels discussed here shares structural similarities with passivity-preserving channels \cite{Alimuddin2020I,Singh2021,Swati2023}, where the fixed point of the dynamics is a passive state. Motivated by this connection, we first explore the role of free energy in characterizing ergodic dynamics. For thermal ergodic channels where fixed points are completely passive states (Gibbs states), free-energy provide a natural framework to understand memory effects. For instance, relative entropy—a widely studied information-theoretic distance measure—has been linked to free energy differences and thermodynamic work in specific regimes \cite{Horodecki2013}. We emphasize how these connections enhance the physical understanding of thermal ergodic channels. For generic ergodic channels, however, where the fixed point may not be a Gibbs state, the association with free energy no longer holds. In such cases, the monotonicity of traditional distance measures like trace distance, or pseudo distance measures like relative entropy persists, but their above thermodynamic relevance vanishes. This shifts the focus to ergotropy—a thermodynamic quantity distinct from free energy that quantifies the maximum extractable work from a system under unitary operations \cite{Pusz1978,Lenard1978,Allahverdyan2004}. Unlike free energy, ergotropy emerges from energy expectations rather than standard distance measures, resulting in unique behavior under ergodic dynamics.

Our analysis reveals that under some specific class of Markovian ergodic channels, the ergotropy of a system decreases monotonically as it approaches the fixed point, eventually reaching a passive state with zero ergotropy. Conversely, non-Markovian ergodic dynamics induces memory effects that lead to ergotropy backflows, causing it to fluctuate over time. This highlights ergotropy as a measure of memory retention in open quantum systems, for a class of ``ergodic" quantum channels. Memory-induced dynamics allow the system to temporarily regain activity, offering a mechanism to retain and utilize ergotropy before it is entirely lost to the environment. Understanding these memory effects has profound implications for optimizing the performance of quantum systems under open dynamics. Specifically, leveraging memory-induced dynamics could enable better retention of ergotropy, enhancing its utility as a thermodynamic resource in quantum batteries. These findings shed light on the interplay between non-Markovianity and environment assisted dynamics, paving the way for a deeper understanding of how memory effects can support the functionality of quantum batteries in practical applications.

\section{\label{sec:level1}generalized ergodic channel with Lindblad evolution} 

In this section we define and construct ergodic quantum operations in arbitrary dimensions and their Lindblad generators respectively. We further elucidate on the ergodic nature of the dynamics, which enables us to define non-ergodicity of a dynamics via a distance measure. Having Lindblad type generators enables us to consider their non-Markovian counterparts and analyse their properties. For a specific example, we restrict ourselves to qubit dynamics with an arbitrary unique fixed point. In the following, we start with the technical preliminaries essential for the study. Before we start, we want to clarify a convention of notation used in this work, that arbitrary quantum states of any dimension are represented by the notation $\zeta$ and for a qubit are denoted by $\rho$. This convention is followed throughout the manuscript. Other states with specific properties are defined accordingly, as we go on. 

\subsection{Technical Preliminaries}
We begin by defining quantum non-Markovianity, which, in this context, refers to quantum channels that are characterized by their indivisibility. The definition of a CP-divisible quantum channel is as follows. 
\begin{definition}\label{def1}
 A quantum channel $\Phi(t,~t_0)[\zeta(t_0)]=\zeta(t)$ is said to be CP-divisible, if it can be realised as the following concatenation for all $t_0 < t' < t$.
\[
\Phi(t,~t_0) = \Phi(t,~t')\circ\Phi(t',~t_0).
\]
where $\Phi(t,~t')$ is a linear CP map. 
\end{definition}
Moreover, if the said dynamics is invertible, then it will correspond to a Lindblad type master equation of the form 

\begin{equation}\label{master1}
\frac{d}{dt}\zeta_t = \mathcal{L}(\zeta_t)= \sum_{i=0}^{n\leq d^2-1} \gamma_i(t)\left[A_i\zeta_tA_i^\dagger-\frac{1}{2}\{A_i^\dagger A_i,~\zeta_t\}\right],
\end{equation}
for a $d$ dimensional system. In such cases the CP divisibility of the dynamics can be determined by the following measure \citep{nm8}
\begin{equation}\label{rhp}
g(t)=\lim_{\delta\rightarrow 0}\frac{||\left(\mathcal{I}\otimes(\mathcal{I}+\delta\mathcal{L})\right)\ket{\phi}\bra{\phi}||_1-1}{\delta},
\end{equation}
with $||A||_1=Tr[\sqrt{A^\dagger A}]$ being the trace norm, $\mathcal{I}$ is identity operation and $\ket{\phi}$ is the maximally entangled state in $d\times d$ dimension. It has been shown \citep{nm8} that, if a quantum operation is divisible, then $\gamma_i(t)\geq 0$ $\forall t~\mbox{and}~i$, giving $g(t)=0$. If the dynamics is indivisible then $\gamma_i(t) < 0$ for at least one $i$ with the consistency condition $\int_0^T \gamma_i(t) dt \geq 0~~\forall T$. 

There is another way to capture indivisible quantum operation by observing information backflow \citep{nm6}. 

\begin{definition}\label{def2}
 Information backflow for indivisible quantum operation is defined by the non-monotonic increment of CP-monotones like trace distance. Mathematically it can be defined by the following condition 
\[
\mathcal{B}(t)= \max_{\zeta_0,\Tilde{\zeta_0}}\frac{d}{dt} D(\zeta_t,\Tilde{\zeta_t}) > 0~ \mbox{for some}~t,
\]
where $D(A,B)=\frac{1}{2}||A-B||_1$ is the trace distance. 
\end{definition}
Note that other CP monotones, in place of trace distance, can also be used in this definition \citep{nm3}. 

\begin{definition}\label{def3}
  The Choi state \citep{choi,jamil} of a given linear map $\mathcal{M}_d(\cdot)$ in $d$-dimension, is given by \[C_{\mathcal{M}_d}=\mathcal{I}\otimes\mathcal{M}_d\left(\ket{\psi_d}\bra{\psi_d}\right),\]
  where $\ket{\psi_d}$ is a maximally entangled state in $d\times d$ dimension. 
\end{definition}
It is also very important to mention that any linear positive map is completely positive, iff the corresponding Choi state is positive \citep{choi,jamil}.

Let us now turn our attention to a specific class of quantum channels and their characteristics. We start with the well known depolarizing channels. These aforementioned channels are unital by nature, having maximally mixed state as their singular fixed point \citep{unital1}. Most famous of the examples is that of qubit Pauli channels given by the master equation 
\begin{equation}\label{master2}
\frac{d}{dt}\rho_t=\sum_{i=1}^3 \Gamma_i(t)\left[\sigma_i\rho_t\sigma_i-\rho_t\right],
\end{equation}
with $\sigma_i$s being the Pauli matrices. Elaborating further, we can say that a completely depolarizing channel for a qubit, taking any arbitrary state $\rho_0$ at initial time, to maximally mixed state, can be represented as 
\[\Lambda_P: \rho\rightarrow \rho' = \sum_{\alpha = 0}^3 \frac{1}{4} \sigma_\alpha \rho_0 \sigma_\alpha,\]
where $\sigma_0=\mathbb{I}$ represents the identity matrix. Note that this map is not in dynamical form, since it does not involve time. We make a dynamical depolarizing map by the following representation 

\begin{equation}\label{depol1}
 \Lambda^{\sigma_0}_t(\rho_0) = p_t\mathcal{I}(\rho_0)+(1-p_t)\Lambda_P(\rho_0),  
\end{equation}

where $0\leq p_t\leq 1$ represents time dependent probability. 

This construction can be generalized for 
a depolarizing map in arbitrary dimensions, allowing for an input state $\zeta_0$ that either matches the channel’s dimension or has a lower dimension but is embedded within the larger-dimensional space. Formally, we represent this as
\begin{equation}\label{depol2}
 \Lambda^{\pi_d}_t(\zeta_0) = p_t\mathcal{I}(\zeta_0)+(1-p_t)\Tilde{\Lambda}_P(\zeta_0),  
\end{equation}
where $\pi_d$ represents the maximally mixed state in d-dimensions and $\Tilde{\Lambda}_P(\cdot)$ represents a completely depolarisng channel in arbitrary dimension, the construction of which is given in appendix \ref{ergoMap-con}.
 These constructions will enable us to study the dynamical properties of such quantum channels, which follows thereafter. But first we will start with the simplest case of qubit dynamics. 

\subsection{\label{sec:level1} Qubit Ergodic Channel}

In the study of stochastic processes, such a process is said to be ergodic, if its statistical traits can be deciphered from a singular sufficiently long realisation of the process. This ergodicity plays a very important role in the study of relaxation to equilibrium. Here we elaborate the preliminary concept of ergodicity from the perspective of quantum information. 

Armed with the aforementioned structure of depolarizing channels, we now define quantum ergodic channels $\Lambda_E(\cdot)$. We start with the simplest case of qubit channels.

The qubit Ergodic channel, defined by $\rho' = \Lambda_E(\rho_0)$, evolves the initial state $\rho_0$ to a predetermined but arbitrary fixed point. 
For a fixed point $\tau = \tau_{00}|0\rangle\langle0| + \tau_{11}|1\rangle\langle 1|$, the channel $\Lambda_E$ follows

\begin{align*}
\Lambda_E: \rho\rightarrow \rho' = \sum_{\alpha=0}^3 A \sigma_\alpha \rho_0 \sigma_\alpha A^\dag ,
\end{align*}

where $A = \sqrt{\frac{\tau}{2}}$. 

Similarly to the depolarizing channel, the dynamical form of the qubit ergodic map is given by 

\begin{align}\label{ergo_map}
 \Lambda_{t}^{\tau}(\rho_0) & =  p_t \mathcal{I}(\rho_0) + (1-p_t) \sum_{\alpha=0}^3 A \sigma_\alpha \rho_0 \sigma_\alpha A^\dag .
\end{align}

The operator $A$ ensures the preservation of a fixed state $\tau$ with a probability of $(1-p_t)$, while the input state is preserved with a probability of $p_{t}$. It is straightforward to verify that $\Lambda_{t}^{\tau}(\tau)=\tau$, i.e. $\tau$ is the fixed point of this dynamics. 

We further investigate the dynamical properties, by deriving the Lindblad type master equation for this map. In Appendix~\ref{sec:appendixLB}, we briefly sketch the method to obtain the Lindblad operator for a dynamical map.
For the particular qubit evolution given in equation \eqref{ergo_map}, its Lindblad master equation can be constructed as

\begin{align}\label{ergo_eq}
\begin{array}{ll}
    \dot{\rho_t} = -\frac{1}{4}\frac{\dot{p_t}}{p_t} \left[ \sigma_3 \rho_t \sigma_3 - \rho_t\right]
    -\tau_{11}\frac{\dot{p_t}}{p_t} \left[ \sigma_- \rho_t \sigma_+ - \frac{1}{2}\{ \sigma_+ \sigma_-, \rho_t\} \right]\\
~~~~~~~~    -\tau_{00}\frac{\dot{p_t}}{p_t} \left[ \sigma_+ \rho_t \sigma_- - \frac{1}{2}\{ \sigma_- \sigma_+, \rho_t\} \right],
\end{array}
\end{align}

where $\sigma_\pm = \frac{1}{2} (\sigma_1 \pm i \sigma_2)$. It is straight forward to check that the state $\tau$ is the fixed point of this dynamics, by putting $\tau$ in place of $\rho_t$. We will then observe that $\tau$ does not change with time. Here it is important to mention that the Lindblad operators are expressed in the basis of the fixed point $\tau$, which we take here to be the computational basis. Even if the fixed point is diagonal in some other arbitrary basis, the corresponding Lindblad operators can be formed in terms of such basis vectors, following the same method; i.e. if $\tau$ is diagonal in some arbitrary basis $\{\ket{\Psi},\ket{\Psi_\perp}\}$, the corresponding Lindblad equation will have the form 
\begin{align}\label{ergo_eq1}
\begin{array}{ll}
    \dot{\rho_t} = -\frac{1}{4}\frac{\dot{p_t}}{p_t} \left[ L_z \rho_t L_z - \rho_t\right]
    -\tau_{11}\frac{\dot{p_t}}{p_t} \left[ L_- \rho_t L_+ - \frac{1}{2}\{ L_+ L_-, \rho_t\} \right]\\
~~~~~~~~    -\tau_{00}\frac{\dot{p_t}}{p_t} \left[ L_+ \rho_t L_- - \frac{1}{2}\{ L_- L_+, \rho_t\} \right],
\end{array}
\end{align}
here the Lindblad operators will be rotated to that basis and represented by $L_z=\ket{\Psi}\bra{\Psi}-\ket{\Psi_\perp}\bra{\Psi_\perp}$, $L_-=\ket{\Psi}\bra{\Psi_\perp}$ and $L_+=\ket{\Psi_\perp}\bra{\Psi}$ respectively, where $\tau=\tau_{00}\ket{\psi}\bra{\psi}+\tau_{11}\ket{\psi_\perp}\bra{\psi_\perp}$.

\subsection{Generalized Ergodic Channel and its Lindblad Construction}

We now generalize our construction to quantum channels in arbitrary dimensions. The qubit Ergodic channel can be extended to \( d \)-dimensional systems. Without loss of generality, this extension assumes a fixed state diagonal in the computational basis, represented as \( \tau_d = \sum_{i} \tilde{\tau}_{ii} |i\rangle \langle i| \).
Building on the generalization of the Pauli channel \citep{nm7}, the action of the generalized ergodic channel can be described by the following equation:

\begin{equation}\label{ergoN1}
   \Lambda_{GE}(\zeta_0)= p_0 (t) \mathcal{I}(\zeta_0)+(1-p_0(t))A_d\Tilde{\Lambda}_p(\zeta_0)A_d^\dagger,
\end{equation}
where $A_d=\sqrt{\frac{\tau_d}{d}}$ and $\Tilde{\Lambda}_p(\cdot)$ is the completely depolarizing channel which takes any given state to identity. It is straightforward to verify that $ \Lambda_{GE}(\tau_d)=\tau_d$; i.e. $\tau_d$ is a fixed point of the dynamics. Now to prove that $\tau_d$ is the sole fixed point, take an arbitrary state $\zeta\neq\tau_d$. So $\Lambda_{GE}(\zeta)=p_0 (t)\zeta+(1-p_0 (t))\tau_d$, which converges to $\tau_d$ iff $p_0 (t)=0$. But for $\zeta$ to be a fixed point of $\Lambda_{GE}$, the condition $\Lambda_{GE}(\zeta)=\zeta$ has to be satisfied for all values of $p_0 (t)$, which is possible only if $\zeta=\tau_d$. Hence it is evident that $\tau_d$ is the only fixed point of $\Lambda_{GE}(\cdot)$.

Based on the previous discussion, we now aim to construct the Lindblad dynamics of an arbitrary dimensional system under the generalized ergodic channel. We outline the dynamics in terms of a master equation that captures the evolution of the system's density matrix under this channel. From the description of the aforementioned $d$-dimenional ergodic map, it is in principle possible to construct the corresponding Lindblad type generator. However, that will involve constructing $d^2\times d^2$  dimensional matrices. Though it is possible to construct such matrices and subsequently the master equation for a given fixed dimension, it is not feasible to find the same for arbitrary dimension. Therefore, we take the approach which is the other way around and infer the possible master equation from the information of the qubit dynamics and then verify that its solution is the aforementioned ergodic map. Based on this inference, we prove the following theorem.

\begin{thm}\label{theorem1}
The dynamics under the generalized ergodic channel $\Lambda_{GE}$ of a dynamical state $\zeta_t=\Lambda_{GE}(\zeta_0)$ in arbitrary dimension, can be expressed by the following sufficient description of a master equation
\begin{align}\label{Nergo}
    \dot{\zeta}_t = \phi_{deph}(\zeta_t) - \frac{\dot{p_t}}{p_t} \sum_{i}\Tilde{\tau}_{ii} \sum_{j\neq i}(|i\rangle\langle j|\zeta_t|j\rangle\langle i| - \frac{1}{2} \{|j\rangle\langle j|,\zeta_t\})
\end{align}
Where the dephasing part $\phi_{deph}(\zeta_t)$ given by
\[\phi_{deph}(\zeta_t) = -\sum_{i, j\neq i} \frac{\Tilde{\tau}_{ii}+\Tilde{\tau}_{jj}}{2} \frac{\dot{p}_t}{p_t} \zeta_{ij} |i\rangle \langle j|\]
accounts for the dephasing of off-diagonal elements of the density matrix $\zeta_t$, where $\zeta_t = \sum_{ij}\zeta_{ij}\ket{i}\bra{j}$ [the subscript 't' in the ``time dependent" components $\zeta_{ij}$ is removed for brevity].
\end{thm}
\begin{proof}
    The proof of this statement requires analysing the master equation given in \eqref{Nergo}.
We start with showing the given master equation for a d-dimensional system can be rearranged to the following form
\[\dot{\zeta}_{ij} = \frac{\dot{p}_t}{p_t} (\zeta_{ij}-\Tilde{\tau}_{ij}),\]
for each $ij$-th element of the density matrix 
\begin{align*}
    \dot{\zeta_{ij}} &= \langle i| \dot{\zeta}_t|j\rangle = \langle i|\phi_{deph}(\zeta_t)|j\rangle \\
    &- \frac{\dot{p_t}}{p_t} \sum_{k}\Tilde{\tau}_{kk} \langle i|\sum_{l\neq k}(|k\rangle\langle l|\zeta_t|l\rangle\langle k| - \frac{1}{2} \{|l\rangle\langle l|,\zeta_t\})|j\rangle .\\
\end{align*}
Let us now consider the diagonal and off-diagonal entries of the density matrix separately.
For the diagonal entries we have 
\begin{equation*}
\begin{array}{ll}
   ~~~~~~~~~~~~ \dot{\zeta_{ii}}=\bra{i}\zeta\ket{i} \\
   \\
    = \bra{i}\phi_{deph}(\zeta_t)\ket{i}- \frac{\dot{p_t}}{p_t}\sum_{k, l\neq k}\Tilde{\tau}_{kk} \langle i|(|k\rangle\langle l|\zeta_t|l\rangle\langle k| - \frac{1}{2} \{|l\rangle\langle l|,\zeta_t\})|i\rangle \\
    \\
    = - \frac{\dot{p_t}}{p_t}\sum_{k, l\neq k}\Tilde{\tau}_{kk} \langle i|(|k\rangle\langle l|\zeta_t|l\rangle\langle k| - \frac{1}{2} \{|l\rangle\langle l|,\zeta_t\})|i\rangle \\
    \\(\mbox{The dephasing will not contribute for the diagonal terms})\\
    \\
    =- \frac{\dot{p_t}}{p_t}\sum_{k, l\neq k}\Tilde{\tau}_{kk}\left(|\langle i|k\rangle|^2 \bra{l}\zeta_t\ket{l}-\frac{1}{2}\bra{i}l\rangle\bra{l}\zeta_t\ket{i}-\frac{1}{2}\bra{l}i\rangle\bra{i}\zeta_t\ket{l}\right)\\
    \\
    =- \frac{\dot{p_t}}{p_t}\sum_{k}\Tilde{\tau}_{kk}\left(|\langle i|k\rangle|^2 - \zeta_{ii} - \left(|\langle i|k\rangle|^2 \zeta_{kk} - \frac{1}{2} \langle i | k \rangle \zeta_{ki} -\frac{1}{2} \langle i | k \rangle \zeta_{ki}\right)\right) \\
    = \frac{\dot{p_t}}{p_t} \left(\zeta_{ii}-\Tilde{\tau}_{ii}\right)
\end{array}
\end{equation*}
For off-diagonal terms,
\begin{align*}
    \dot{\zeta}_{ij,(i \neq j)} =&  \langle i|\phi_{deph}(\zeta_t)|j\rangle- \frac{\dot{p_t}}{p_t} \sum_{k=0}^{d-1}\Tilde{\tau}_{kk} \sum_{l\neq k}\bigg(\langle i|k\rangle\langle l|\zeta_t|l\rangle\langle k|j\rangle \\
    &- \frac{1}{2} (\langle i|l\rangle\langle l|\zeta_t|j\rangle + \langle i|\zeta_t|l\rangle \langle l|j\rangle) \bigg) \\
    =&  \langle i|\phi_{deph}(\zeta_t)|j\rangle- \frac{\dot{p_t}}{p_t} \sum_{k=0}^{d-1}\Tilde{\tau}_{kk} \bigg(0 - \frac{1}{2} \bigg(\langle i|\zeta_t|j\rangle + \langle i|\zeta_t|j\rangle \\
    & - \langle i|k\rangle\langle k|\zeta_t|j\rangle - \langle k|j\rangle\langle i|\zeta_t|k\rangle \bigg) \bigg) \\
    =& \langle i|\phi_{deph}(\zeta_t)|j\rangle + \frac{\dot{p_t}}{p_t} [\zeta_{ij} - \frac{1}{2}(\Tilde{\tau}_{ii} + \Tilde{\tau}_{jj}) \zeta_{ij}] \\
    =& \frac{\dot{p_t}}{p_t} \zeta_{ij}= \frac{\dot{p_t}}{p_t} (\zeta_{ij} - \Tilde{\tau}_{ij})~~(\mbox{since}~~ \Tilde{\tau}_{ij, (i \neq j)} =0).
\end{align*}
Therefore from the analysis of both diagonal and off-diagonal terms it is shown that the $ij$-th components of the density matrix obeys the equation $\dot{\zeta}_{ij} = \frac{\dot{p}_t}{p_t} (\zeta_{ij}-\Tilde{\tau}_{ij})$. If we now solve the differential equation considering the initial probability $p_0=1$, then we find the original ergodic map $p_t\zeta_0+(1-p_t)\tau_d$ given in \eqref{ergo_map}. This completes the proof.
\end{proof}

\subsection{Non-Markovianity Analysis of Ergodic Channels} \label{SecD}

In this subsection, we are now going to analyse the aspect of non-Markovianity of the ergodic channels. We will do so by considering both the phenomena of divisibility and information backflow of a given channel. For a clear understanding, we start with the qubit ergodic channel as an example expressed by the master equation given in \eqref{ergo_eq}. Following the divisibility based non-Markovianity measure, which we can call the RHP (Rivas-Huelga-Plenio) measure given in \eqref{rhp}, we can explicitly calculate the non-Markovianity of the aforementioned qubit channel given in \eqref{ergo_eq}. The Choi state (Def. \ref{def3}) for this particular dynamics, over which this measure will be calculated  is given by 
\begin{equation}
\mathcal{C}_\delta=\mathcal{I}\otimes(\mathcal{I}
+\delta\mathcal{L})\ket{\phi}\bra{\phi},
\end{equation} 
where $\mathcal{L}$ is the Lindblad type superoperator given in \eqref{ergo_eq}, $\delta$ is an infinitesimally small time period and $\ket{\phi}=\frac{1}{\sqrt{2}}(\ket{00}+\ket{11})$ is the maximally entangled states in $2\times 2$ dimension. Based on these, the RHP measure can be calculated as 
\[g(t)= \frac{3}{2}\left|-\frac{\dot{p}_t}{p_t}\right|,\]
when $\dot{p}_t$ is positive. When the same is negative, the dynamics is divisible and hence Markovian, giving $g(t)=0$ for all $t$. Note that in our calculation, we have neglected $2^{nd}$ and higher order terms containing $\delta$, since they are infinitesimally small. Here it is important to note that the non-Markovianity measure does not depend on the fixed point $\tau$. Therefore irrespective of whether the operation is depolarizing with $\tau = \mathbb{I}/2$, amplitude damping with 
$\tau = \ket{0}\bra{0}$ or any other ergodic dynamics with corresponding arbitrary fixed point $\tau$, the non-Markovianity will solely depend on the time evolution characteristics of the dynamic probability $p_t$. As an example, we consider an exponentially decaying dynamic probability $p_t = e^{-\gamma t}$, where $\gamma$ is a constant factor. By analysing the map given in \eqref{ergo_map}, we can directly deduce that given an exponentially decaying dynamical probability, any given initial state will monotonically converge to the given fixed point, as expected for an evolution that is Markovian in nature. It is also strightforward to deduce that such exponentially decaying dynamic probability will indicate a divisible and hence Markovian evolution, since $\dot{p}_t$ is always negative. On the otherhand, if we consider $p_t=\cos^2 (\omega t)$, the RHP measure gives us $g(t)=2\omega\left|\tan\omega t\right|$, where $ t = \frac{(2n+1)\pi+\theta}{2\omega}$ with $0<\theta<\pi$ and $n=(0,1,2,...)$. Note that the aforementioned results on RHP measure are restricted to qubits. However, in order to demonstrate a more general result on non-Markovianity analysis of arbitrary dimensional ergodic evolutions, we exploit another prominent measure of non-Markovianity in the following discussion.  

The other established quantification of non-Markovianity is done through quantifying the information backflow, as described by Definition \ref{def2}, which we abbreviate as BLP (Breuer-Laine-Piilo) measure. To calculate this measure, we consider two arbitrary qubits $\rho_t=\frac{1}{2}\left(\mathbb{I}+\Vec{\chi(t)}.\Vec{\sigma}\right)$ and $\mu_t=\frac{1}{2}\left(\mathbb{I}+\Vec{\Pi(t)}.\Vec{\sigma}\right)$, with $\Vec{\sigma}=\{\sigma_1,\sigma_2,\sigma_3\}$, $\Vec{\chi(t)}=\{\chi_1(t),\chi_2(t),\chi_3(t)\}$, $\Vec{\Pi(t)}=\{\Pi_1(t),\Pi_2(t),\Pi_3(t)\}$ obeying the well known constraints of positive semi-definite matrices, with unit trace. The time derivative of trace distance between $\rho_1(t)$ and $\rho_2(t)$ can therefore be calculated as  
\[
\begin{array}{ll}
~\frac{\dot{p}_t}{p_t}\sqrt{\left(\chi_1(t)-\Pi_1(t)\right)^2+\left(\chi_2(t)-\Pi_2(t)\right)^2+\left(\chi_3(t)-\Pi_3(t)\right)^2},\\
\\
=\dot{p}_t\sqrt{\left(\chi_1(0)-\Pi_1(0)\right)^2+\left(\chi_2(0)-\Pi_2(0)\right)^2+\left(\chi_3(0)-\Pi_3(0)\right)^2}.
\end{array}
\]
Here we use the fact that $\rho_i(t)=p_t\rho_i(0)+(1-p_t)\tau$. Therefore the BLP measure of non-Markovianity, which is the maximum of the previous quantity over the initial states, can be evaluated as 
\[
\mathcal{B} (t) = \frac{dp_t}{dt},
\]
when $\dot{p}_t>0$, giving rise to information backflow. It is interesting to note here that both RHP and BLP measures are equivalent for the case of qubit ergodic evolution. From their expressions we observe that $\mathcal{B}(t)$ will give non zero value indicating information backflow, at the same instant when $g(t)$ is non zero, when $\dot{p}_t > 0$. 

Here, we explicitly demonstrate that the results regarding BLP measure established for qubit ergodic channels can be extended to channels of arbitrary dimensions. In case of divisibility, we know from previous literature \citep{nm8} that if one or more of the Lindblad coefficients are negative at any time of the dynamics, it is not CP-divisible. Following this, we can see from the general master equation given in \eqref{Nergo}, that the evolution is not CP-divisible, when $\dot{p}_t > 0$. For the case of information backflow, for a general ergodic evolution of d-dimensional system given by $\zeta_0 [\mbox{or}~ \Tilde{\zeta_0}]\rightarrow p_t\zeta_0 [\mbox{or}~ \Tilde{\zeta_0}]+(1-p_t)\tau$, the trace distance between two arbitrary states $\zeta_t$ and $\Tilde{\zeta_t}$ at any arbitrary time $t$, can be expressed as 
\[ D\left( \zeta_t-\Tilde{\zeta_t}\right) 
= p_t D\left( \zeta_0-\Tilde{\zeta_0}\right). \]
From the expression, using the fact that the trace distance is upper bounded by unity, the BLP measure can be derived to be $\mathcal{B}_d(t)=\dot{p}_t$. Hence information backflow will always occur with breaking of CP-divisibility, when $\dot{p}_t > 0$.

Note that the aspect of divisibility and non-Markovianity is much more rigorous \cite{wolf2008dividing} and hence needed further scrutiny for the ergodic channels introduced here. In appendix \ref{sec:appendixDiv} we present a detailed analysis of divisible ergodic channels.

\section{Ergodic Channels and Their Broader Physical Implications}

In the previous sub-Section \ref{SecD}, we analyzed the non-Markovian behavior of ergodic channels by virtue of divisibility \citep{nm8} and information backflow \citep{nm6} criteria. Extending this exploration, we now turn to the $\alpha$- relative Rényi entropies, a family of pseudo distance-based measures, to deepen our understanding of the physical implications of ergodic channels. The $\alpha$- relative Rényi entropies \citep{datta}, defined for a pair of quantum states \((\zeta_1, \zeta_2)\), given by
\[
D_\alpha \left(\zeta_1||\zeta_2\right) = \frac{1}{\alpha - 1} \ln \left( \zeta_1^\alpha \zeta_2^{1 - \alpha} \right),
\]
where \(\alpha \in [0, \infty)\). Although this measure does not satisfy all formal mathematical criteria for a distance measure, it has significant physical relevance across various fields, including the resource theory of thermodynamics \cite{brandao}, uncertainty relations \cite{uncer1}, and other branches of quantum information theory \citep{gilad1,gilad2,Holevo_2012}. For ergodic channels, if \(\tau_d\) represents the $d-$ dimensional fixed point, the $\alpha$- relative Rényi entropies \(D_\alpha(\zeta_1 || \tau_d)\) become monotonic for any given ergodic channel, where $\zeta_1$ represents an arbitrary initial state. This is straightforward to prove, by using the monotonicity of the $\alpha$- relative Rényi entropies under completely positive trace preserving maps \cite{lieb} and the fact that $\tau_d$ is the fixed point under corresponding ergodic evolution. To associate greater physical implications we now turn to more realistic ergodic channels.

\subsection{Thermal Ergodic Channels}

An ergodic channel is referred to as a thermal ergodic channel when its fixed point \(\tau_d\) assumes the Gibbsian form \(\tau_\beta = \frac{\exp(-\beta H)}{\operatorname{Tr}[\exp(-\beta H)]}\), where \(\beta = 1/k_B T\) is the inverse temperature. In this context, the \(\alpha\)-relative Rényi entropies serve as generalized free energy differences \cite{Petz1986, Donal1987, Horodecki2013}. Specifically,
\[
F_\alpha(\zeta)-F_\alpha(\tau_\beta) = \frac{1}{\beta}D_\alpha(\zeta||\tau_\beta),
\]
with \(\alpha \in [0, \infty)\), gives rise to an infinite series of second laws of thermodynamics within the finite quantum regime \cite{Brandao2013,brandao,Horodecki2013,Aberg2013}. In the asymptotic (classical) limit, these generalized free energies turn out to be a single quantity, and one recovers the standard second laws of thermodynamics.

More rigorously, in the limit \(\alpha \to 1\), the $\alpha$- relative Rényi entropies reduce to the usual quantum relative entropy,
\[
D(\zeta||\tau_\beta) = \tr\left[\zeta(\ln \zeta - \ln \tau_\beta)\right],
\]
which gives the usual free energy differences as,
\begin{align}\label{thermodynamic work}
    W_{th}(\zeta) &= F(\zeta) - F(\tau_\beta) = \tr(\zeta H) - T S(\zeta) \nonumber\\
    &= \frac{1}{\beta}D(\zeta||\tau_\beta).
\end{align}
This quantity, \(W_{th}(\zeta) = F(\zeta) - F(\tau_\beta)\), represents the maximum extractable work from the system \(\zeta\) when in contact with a thermal bath at inverse temperature \(\beta\), achievable in the asymptotic limit \cite{Brandao2013}.

In contrast, at \(\alpha \to \infty\), the $\alpha$- relative Rényi entropies capture the max-relative entropy, with
\[
W_{\infty}(\zeta) = \frac{1}{\beta}D_{\infty}(\zeta || \tau_\beta) = \frac{1}{\beta}\inf\{\gamma : \zeta \leq 2^\gamma \tau_\beta\},
\]
which denotes the single-shot work extractable from the state \(\zeta\). In the the end, for \(\alpha \to 0\),
\[
W_0(\zeta) = \frac{1}{\beta}D_0(\zeta || \tau_\beta) = -\frac{1}{\beta}\ln \left[\tr(\Pi_{\zeta} \tau_\beta)\right],
\]
where \(\Pi_{\zeta}\) is the projection onto the support of \(\zeta\), quantifying the thermodynamic formation cost of the state \(\zeta\). These quantities follow the ordering \(W_0(\zeta) \leq W_{th}(\zeta) \leq W_{\infty}(\zeta)\), with both inequalities converging to equalities in the asymptotic limit.

\begin{figure*}[ht]
    \centering
    \begin{subfigure}[b]{0.48\textwidth}
        \centering
        \includegraphics[width=\textwidth]{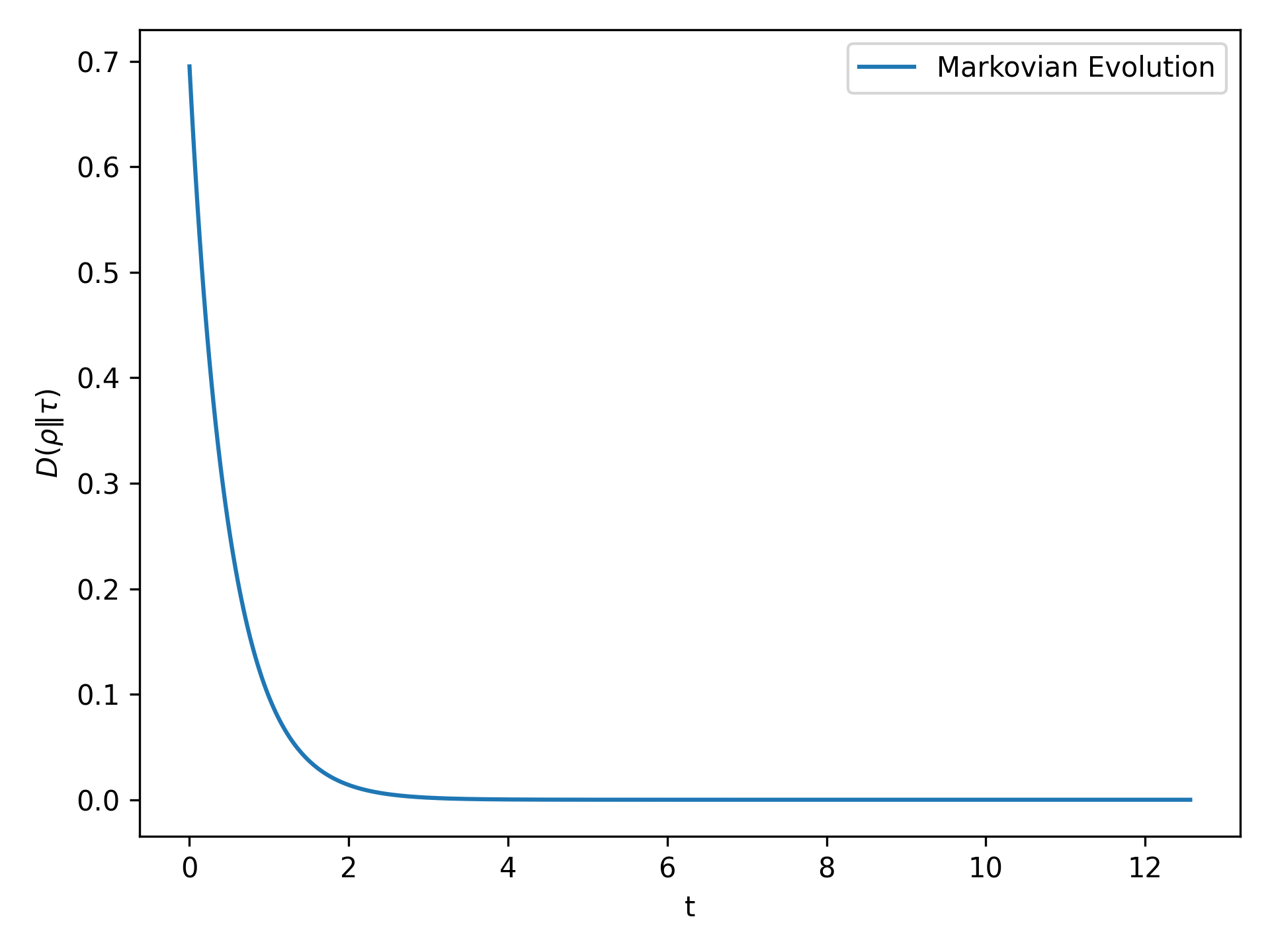}
        \caption{}
        \label{fig:qre_exp}
    \end{subfigure}
    \hfill
    \begin{subfigure}[b]{0.48\textwidth}
        \centering
        \includegraphics[width=\textwidth]{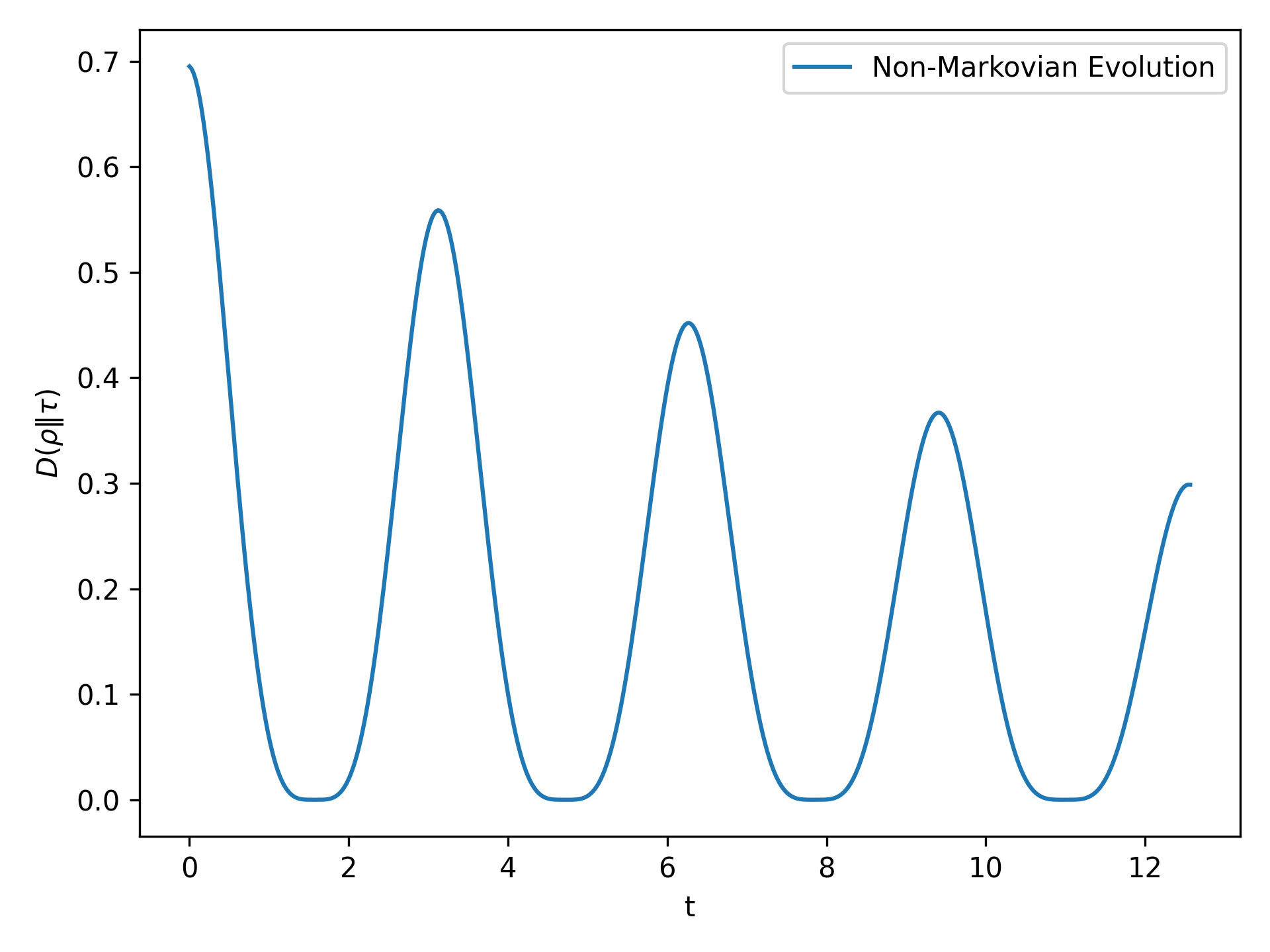}
        \caption{}
        \label{fig:qre_cos}
    \end{subfigure}
    \caption{\it Change in Relative Entropy under Markovian and Non-Markovian Evolution-- The initial state is represented in the Bloch sphere as \( \vec{r}_{\text{init}} = \left(\frac{1}{2}, 0, \frac{1}{2}\right) \), while the thermal state is given by \( \vec{r}_{\tau} = (0, 0, \frac{1}{2}) \). The thermal state is defined in the computational basis \(\{\ket{0}, \ket{1}\}\), with temperature \(T = 1\) and the corresponding Hamiltonian \(H_s = e \ket{1}\bra{1}\), where the energy gap \(e = 1.099\) is expressed in natural units (\(\hbar = 1\), \(k_B = 1\)). Fig. 1(a) illustrates the Markovian dynamics, showing a monotonic decrease in relative entropy (or thermodynamically extractable work) over time due to the continuous dissipation of resources into the bath. In contrast, Fig. 1(b) depicts non-Markovian dynamics, where the environment facilitates backflow of previously lost resources, retaining memory during the evolution and eventually leading to saturation.}
    \label{RelEnt}
\end{figure*}

Since we are considering thermal channels, all discussed quantities exhibit monotonic behavior under Markovian dynamics, with their resource content progressively decreasing as the system moves toward equilibrium. However, as illustrated in Fig. (\ref{RelEnt}), where we plot the time evolution of relative entropy (Eq. \ref{thermodynamic work}), non-Markovian dynamics reveal a contrasting non-monotonic behavior for usual relative entropy, underscoring the influence of memory effects on resource retention in terms of thermodynamically extractable work.

When extending our focus from thermal channels to arbitrary ergodic channels with fixed points that are non-thermal states, significant distinctions arise. In these cases, the usual thermodynamic interpretations of extractable work or free energy, typically associated with a thermal state, cannot be directly considered. This is because free energy and extractable work are defined relative to a thermal equilibrium state, which is not well-defined for a non-thermal fixed point. While the notion of distance, such as relative entropy, still holds, it no longer directly corresponds to extractable work.Nonetheless, the \(\alpha\)-relative Rényi entropies retain their monotonicity, provided the distance measures are defined with respect to the fixed point of the Markovian ergodic channel.

In the next section, we explore how a thermodynamic interpretation can still be assigned to ergodic channels whose fixed points are passive states. For systems with dimensions greater than two, the set of passive states extends beyond the set of thermal states. Unlike thermal states, general passive states lack the concept of a unique temperature since their probability distributions across energy eigenstates deviate from the Gibbsian distribution, marking them as non-equilibrium states \cite{Skrzypczyk2015}. Nevertheless, passive states retain thermodynamic relevance: under isolated (unitary) evolution, they represent the lowest energetic configurations from which no work can be extracted, much like thermal states. We will examine in detail how the work extractable from an isolated quantum system, known as ergotropy, exhibits monotonicity under ergodic channels when passive states (need not to be thermal states) serve as the fixed point.

\subsection{\label{sec:level3} Passive State Preserving Ergodic Channels
}

Ergotropy quantifies how much optimal work one is able to extract from an isolated quantum system (battery) \cite{Pusz1978, Lenard1978, Allahverdyan2004, Skrzypczyk2015, Salvia2020}. As the system is isolated, its dynamics are governed solely by unitary operations and extractable work from an arbitrary system \(\zeta\) reads as
\begin{align}
    W_e(\zeta) = E(\zeta) - \min_U E(U \zeta U^\dagger),
\end{align}
where \(E(X) := \tr(XH)\) denotes the average energy of the system, and \(H = \sum_i \epsilon_i |i\rangle \langle i|\) represents the system Hamiltonian with \(\epsilon_i \leq \epsilon_{i+1}\). Here, \(\epsilon_i\) are the energy eigenvalues with the corresponding energy eigenvectors \(|i\rangle\). The minimization of unitary operations is required to evolve the system towards the lowest energetic state. For a given system \(\zeta\) with the spectral decomposition \(\{\lambda_i\}\) arranged in non-increasing order, the state with the lowest energy, called the passive state, is expressed as $\zeta_p = \sum_i \lambda_i |i\rangle\langle i|.$ The final state is called {\it passive} state as it yields zero ergotropy under the unitary evolution.
\par
Recently, ergotropy has attracted significant attention for its role in characterizing various features of quantum correlations \cite{Mukherjee2016, Francica2017,Alimuddin2019,Francica2020,Sen2021,Touil2022,Francica2022,Puliyil2022,Tirone2023,Yang2024,Sun2024} and its important implications for quantum batteries \cite{Alicki2013,Andolina2019,Alimuddin2020S,Cruz2022,Joshi2022,Yang2023,Dominik2023,Joshi2024,Campaioli2023,Halder2024}. As a thermodynamic resource, ergotropy has led to the development of a resource-theoretic framework \cite{Alimuddin2020I,Singh2021,Swati2023}, though defining a complete set of free operations within this theory remains a challenge. Notably, a specific class of ergodic channels has been identified where ergotropy exhibits monotonicity. However, not all ergodic channels share this property, even when their fixed points correspond to passive states. A key example of this is the Gaussian-preserving channel, which, despite being a thermal-state-preserving ergodic channel, does not exhibit the monotonicity of ergotropy \cite{Faist2015}. This is because such a channel can induce coherence in a diagonal passive state, transforming it into an active state, and thereby violating monotonicity. Moreover, it is well-established that ergotropy, thermodynamic work, entropy, and single-shot extractable work are independent quantities, underscoring ergotropy as a genuine resource \cite{Alimuddin2020I, Allahverdyan2004}.

One of the nontrivial challenges is linking ergotropy with a specific distance measure, unlike thermal work \(W_{th}\). This remains an open question, partly due to the complex transformation of ergotropy during state evolution. When a state \(\zeta\) evolves to \textcolor{red}{\(\chi\)}, the corresponding passive states \(\zeta_p\) and \textcolor{red}{\(\chi_p\)} also change (both of these changes could be arbitrarily connected), complicating the understanding of how ergotropy behaves under these transformations. This challenge can be further understood through the framework introduced in \cite{Francica2020}, where the total ergotropy of a state \(\rho\) is partitioned as \(W_e(\rho) = \mathcal{E}_c(\rho) + \mathcal{E}_i(\rho)\). Here, \(\mathcal{E}_c(\rho)\), referred to as the coherence contribution, is associated with distance measures and is a function of the coherence in the state and relative entropy. On the other hand, \(\mathcal{E}_i(\rho) = E(\rho) - \max_{U_i} E(U_i \rho U_i^\dagger)\) represents the incoherent contribution, where \(U_i\) denotes unitaries that preserve the coherence of the state. Unlike \(\mathcal{E}_c(\rho)\), the term \(\mathcal{E}_i(\rho)\) involves the Hamiltonian of the system and considers observational quantities, making it challenging to associate with any specific distance measure. Further insights into this issue are provided in \cite{Lobejko2021,Sparaciari2017}, where the authors establish relationships between ergotropy, free energy, and distance measures under finite and asymptotic scenarios. These analyses involve auxiliary systems such as a bath and a work reservoir, constrained by natural conditions on the allowed unitary transformations, including (i) energy conservation and (ii) translational invariance. While this framework broadens the conceptual understanding of ergotropy, it remains nontrivial to identify a comprehensive set of channels that satisfy ergotropy monotonicity or to pinpoint a universal distance function directly associated with the ergotropy of an arbitrary system—particularly without any auxiliary assistance. The key message here is that if ergotropy could be rigorously connected to a specific distance measure, it would significantly simplify Markovian analyses. Specifically, monotonicity under completely positive (CP) maps could be directly leveraged, streamlining the evaluation of ergotropy evolution in quantum channels. In the absence of such a direct association, however, the problem remains both intricate and crucial for advancing our understanding of ergotropy's role in quantum thermodynamics.

 In this article, we primarily explore the dynamics of ergotropy for a known class of passive state preserving ergodic channels and demonstrate how information backflow, due to non-Markovianity, affects the monotonic nature of ergotropy in the system. This analysis provides new insights into the interplay between ergotropy and ergodic channels, particularly in non-Markovian environments. We start with the following Lemmas, which characterizes the class of quantum channels that demonstrate the monotonicity of the ergotropic quantity.
\begin{Lemma}\cite{Alimuddin2020I}\label{lemma1}
Ergotropy is a monotonic quantity for any arbitrary energy-preserving quantum channels.
\end{Lemma}
\begin{proof}
The set of energy-preserving channels (EPC) defined by \(\Lambda_{EPC}\) are unital, as they preserve the identity, i.e., \(\Lambda_{EPC}(I)=I\) \cite{Chiribella2017}. If a transformation \(\zeta \mapsto \Tilde{\zeta}\) is governed by any unital map, then the system \(\zeta\) majorizes \(\Tilde{\zeta}\) \cite{Chiribella2017U}. Majorization is a general concept used to quantify the disorderness of a system, and if a system \(\zeta\) majorizes \(\Tilde{\zeta}\) (denoted \(\zeta \succ \Tilde{\zeta}\)), then the entropic relation \(S(\zeta) \leq S(\Tilde{\zeta})\) holds. Interestingly, the corresponding passive state energies also follows this order: \(E(\zeta_p) \leq E(\Tilde{\zeta}_p)\) if the majorization relation holds \cite{Marshall2011}. Therefore, if \(\Lambda_{EPC}(\zeta) = \Tilde{\zeta}\), then \(W_e(\zeta) \geq W_e (\Tilde{\zeta})\), as the transformation is energy-preserving. This completes the proof.
\end{proof}
This proof can be extended even to the class of energy non-increasing unital channels to exhibit ergotropic monotonicity. However, in the next lemma, we will demonstrate some ergotropy non-increasing channels that need not be unital. In \cite{Swati2023}, the authors define a channel:
\begin{equation}\label{EPCPR}
    \Lambda_{EPCPR} (\zeta) = p \Lambda_{EPC} (\zeta) + (1-p) \tau_p,
\end{equation}
where `EPCPR' stands for ``energy-preserving channels and passive resets," and \(\tau_p\) is an arbitrary passive state where the resetting occurs with probability \(1-p\).
\begin{Lemma}\label{lemma2}\cite{Swati2023}
    Ergotropy is monotone under $\Lambda_{EPCPR}$ channels.
\end{Lemma}
\begin{proof}
   Ergotropic work is a convex function, i.e., \(W_e(\sum_i p_i \zeta_i) \leq \sum_i p_i W_e(\zeta_i)\). This implies that if we have knowledge of the individual states \(\zeta_i\) that compose a given mixture, we can extract at least as much work as we can from the mixed state itself. Utilizing this fact and the expression for the EPCPR channel (Eq. \ref{EPCPR}), we can demonstrate that:
\begin{align}
    W_e(\Lambda_{EPCPR}(\zeta)) &\leq p W_e(\Lambda_{EPC}(\zeta)) + (1-p) W_e(\tau_p) \nonumber \\
    &\leq p W_e(\Lambda_{EPC}(\zeta)) \nonumber \\
     &\leq W_e(\Lambda_{EPC}(\zeta)) \nonumber \\
     &\leq W_e(\zeta).
\end{align}
The last inequality is derived using Lemma \ref{lemma1}, thus completing the proof.
\end{proof}
The set of ergodic channels considered in this article has the form \(\Lambda_{\tau_d} (\zeta) = p_t \zeta + (1-p_t) \tau_d\) which is quite similar to the EPCPR channel, where \(\Lambda_{EPC}\) is treated as an identity channel. However, a significant difference arises from the state \(\tau_d\), which does not necessarily have to be a passive state. However we can always, in principle, engineer the system Hamiltonian such that \(\tau_d\) becomes a passive state (later on denoted by $\tau_p$ only). In that case, the channels become quite similar, with the only difference being the time-dependent probability \(p_t\) due to the dynamical evolution of the ergodic channel. The following lemma is essential to exhibit the monotonicity of ergodic channels. In the subsequent lemma and theorem, we demonstrate how CP-divisibility can lead to this monotonicity throughout the dynamics of the ergodic channel.

\begin{Lemma}\label{lemma3}
    For a pair of EPCPR channels \(\Lambda^p_{\tau_p}(\zeta)=p\zeta+(1-p)\tau_p\) and \(\Lambda^q_{\tau_p}(\zeta)=q\zeta+(1-q)\tau_p\), \(W_e(\Lambda^p_{\tau_p}(\zeta)) \geq W_e(\Lambda^q_{\tau_p}(\zeta))\) if and only if \(p \geq q\).
\end{Lemma}
\begin{proof}
    Define an EPCPR channel:
    \begin{align}
        N^r_{\tau_p}(\zeta) = r \zeta + (1-r) \tau_p.
    \end{align}
    We can always find a suitable value of probability \(r\) such that the action of channel \(\Lambda^q_{\tau_p} (\zeta)\) can be obtained for any arbitrary density operator by the composition \(N^r_{\tau_p} \circ \Lambda^p_{\tau_p} (\zeta)\). In that case, we need to choose \(r = \frac{q}{p}\). As the channels are EPCPR, using Lemma \ref{lemma2}, we can directly state that
    \[
    W_e(\zeta)\geq W_e(\Lambda^p_{\tau_p}(\zeta)) \geq W_e(\Lambda^q_{\tau_p}(\zeta)) \quad \text{if} \quad 1 \geq p \geq q.
    \]

    To prove the ``only if" part, we proceed by contradiction. Assume:  
    \[
    W_e(\Lambda^p_{\tau_p}(\zeta)) \geq W_e(\Lambda^q_{\tau_p}(\zeta)) \Rightarrow q \geq p,
    \]
    then negating the assumption gives:
    \[
    \neg \left(W_e(\Lambda^p_{\tau_p}(\zeta)) \geq W_e(\Lambda^q_{\tau_p}(\zeta))\right) \Rightarrow p \geq q,
    \]
    which contradicts the ``if" part. Therefore, we conclude:
    \[
    W_e(\Lambda^p_{\tau_p}(\zeta)) \geq W_e(\Lambda^q_{\tau_p}(\zeta)) \Rightarrow p \geq q.
    \]
\end{proof}
\begin{thm}\label{theorem2}
    CP-divisible quantum operations of the form: $\Lambda_{\tau_p}(\zeta) = p_t\zeta_0+(1-p_t)\tau_p,$ retain the monotonicity of ergotropy. 
\end{thm}
\begin{proof}
    From the above discussion it is apparent that under operations of the following form, ergotropy is a monotonically decreasing quantity. 
\[
\Lambda_{\tau_p}(\zeta) = p_t\zeta_0+(1-p_t)\tau_p,
\]
with the fixed point $\tau_p$ being a passive state. The Lindblad type master equation for the passivity preserving channels is given in \eqref{Nergo} with the specific case of the fixed point $\tau_p$ being a passive state. From the proof of Theorem \ref{theorem1}, it is known that the dynamic evolution of the density matrix $\zeta_t$ for such channels can be expressed as 
\[
\dot{\zeta}_t = \frac{\dot{p}_t}{p_t}\left(\zeta_t-\tau_p\right).
\]
We also know from the proof of Theorem \ref{theorem1} that the solution of this equation gives rise to the channel $\Lambda_{\tau_p}(\cdot)$. If we now consider an intermediate map from an arbitrary time $t$ to $t+\delta$, we see that the density matrix takes the form 
\[
\begin{array}{ll}
\lambda_{\tau_p}^\delta(\zeta_t)=\zeta_{t+\delta}= \zeta_t+\delta\frac{d\zeta_t}{dt}\\
\\
\mbox{(Here we take taylor expansion upto 1st order)}\\
\\
= \zeta_t+\delta\frac{\dot{p}_t}{p_t}\left(\zeta_t-\tau\right)
= \left(1+\delta\frac{\dot{p}_t}{p_t}\right)\zeta_t-\delta\frac{\dot{p}_t}{p_t}\tau_p\\
\\
= \left(p_t+\delta \dot{p}_t\right)\zeta_0 +\left(1-p_t-\delta \dot{p}_t\right)\tau_p\\
\\
= p_{t+\delta}\zeta_0+(1-p_{t+\delta})\tau_p
\end{array}
\]
So the evolution $\lambda_{\tau_p}^\delta(\cdot)$ retains the form of $\Lambda_{\tau_p}$ if $p_{t+\delta}$ is a probability for every value of $p_t$. That can only be possible if $\dot{p}_t < 0$ $\forall t$ (here we use the condition that $\delta$ is so small that $\delta\dot{p}_t << 1$). This in fact is the condition for CP-divisibility for such ergodic channels. It is clear that for any time \(t\) to \(t+\delta\), the probability fraction over the initial state \(p_t\) changes to \(p_{t+\delta} \leq p_t\). Using Lemma \ref{lemma3}, we can conclude that \(W_e(\Lambda^{p_{t+\delta}}_{\tau_p}(\zeta)) \leq W_e(\Lambda^{p_t}_{\tau_p}(\zeta))\). Hence, the theorem is proven.
\end{proof}

\begin{figure*}[ht]
    \centering
    \begin{subfigure}[b]{0.48\textwidth}
        \centering
        \includegraphics[width=\textwidth]{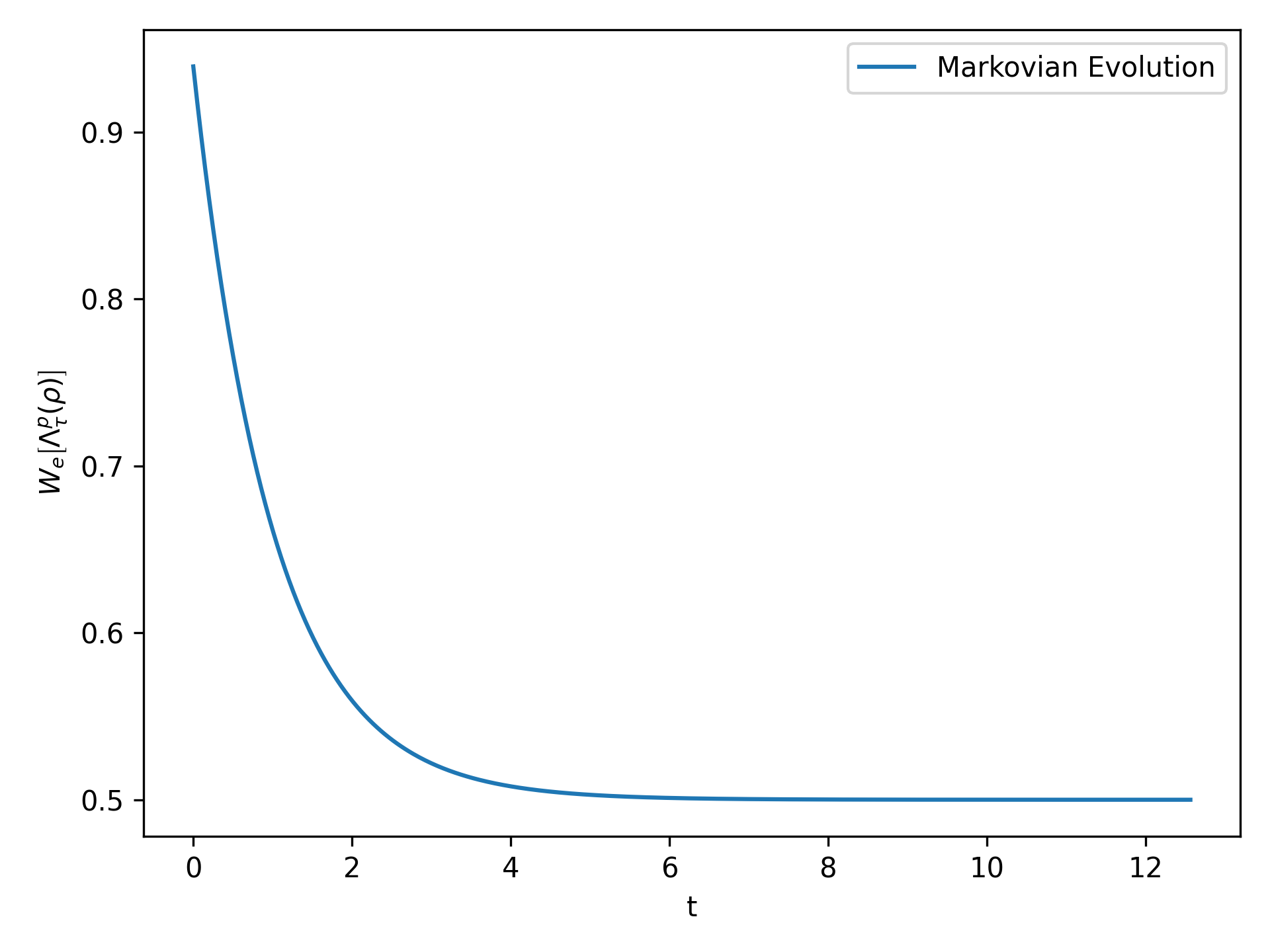}
        \caption{}
        \label{fig:ergo_exp}
    \end{subfigure}
    \hfill
    \begin{subfigure}[b]{0.48\textwidth}
        \centering
        \includegraphics[width=\textwidth]{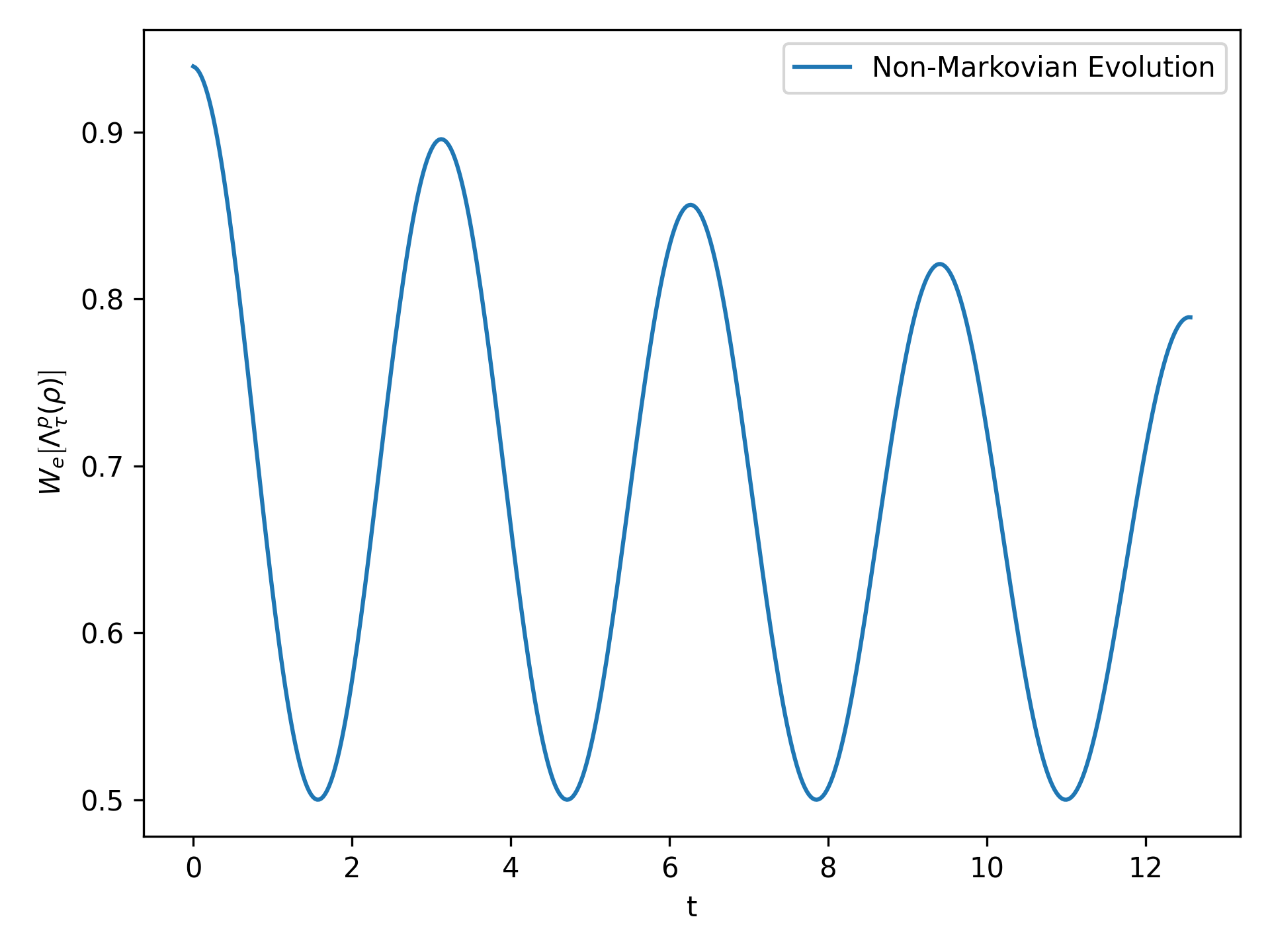}
        \caption{}
        \label{fig:ergo_cos}
    \end{subfigure}
    \caption{{\it Change of ergotropy under Markovian and non-Markovian evolution.--} 
The initial state is chosen as \( \ket{\psi} = \alpha \ket{0} + \beta \ket{\phi} + \gamma \ket{\phi^\perp} \), a state in the qutrit system, with coefficients \( (\alpha, \beta, \gamma) = \left(\sqrt{\frac{1}{4}}, \sqrt{\frac{1}{2}}, \sqrt{\frac{1}{4}}\right) \), where \( \ket{\phi} = \frac{\ket{1} + \ket{2}}{\sqrt{2}} \) and \( \ket{\phi^\perp} = \frac{\ket{1} - \ket{2}}{\sqrt{2}} \). The passive state is defined as \( \tau = \frac{1}{2} \ket{0}\bra{0} + \frac{1}{4} \ket{1}\bra{1} + \frac{1}{4} \ket{2}\bra{2} \), which does not correspond to any thermal state under the given Hamiltonian \( H_s = \ket{1}\bra{1} + 2\ket{2}\bra{2} \). Note that this case of non-thermal passive state arrives as the analysis moves to the qutrit dimension, whereas for the qubit case any passive state corresponds to a thermal state. The qutrit passive dynamics used in this case is derived from equation \eqref{ergoN1} with $\tau$ being the sole fixed point. In Fig. 2(a), we examine the monotonic decay of ergotropy under the Markovian EPCPR channel, consistent with the loss of extractable work over time. In contrast, Fig. 2(b) illustrates the fluctuation of ergotropy under non-Markovian dynamics, highlighting memory retention and information backflow during the evolution.}

    \label{fig:ergotropy}
\end{figure*}

\begin{cor}\label{cor1}
For the EPCPR operation, if it is not CP-divisible, the ergotropy will increase at all those times, when CP-divisibility breaks down.    
\end{cor}

\begin{proof}
This is a consequence of Lemma \ref{lemma3} and Theorem \ref{theorem2}. From Theorem \ref{theorem2}, it is apparent that  CP-divisibility of EPCPR channels will break down, when $\dot{p_t} > 0$, i.e. $p_{t+\delta}>p_t$; which means the instantaneous state goes away from the passive fixed point $\tau$. It is therefore natural that ergotropy will increase when $\dot{p_t} > 0$. This fact can be proved mathematically by virtue of Lemma \ref{lemma3}. We can use Lemma \ref{lemma3} to state that for all those $t,~\mbox{and}~\delta$, for which $\dot{p_t} > 0$, we have $p_{t+\delta}>p_t$, and hence ergotropy at $t+\delta$ is greater than ergotropy at $t$. This proves the corollary. 
\end{proof}

Next, we show how the ergotropy of an arbitrary qubit state $\rho$ changes through a channel $\Lambda^p_{\tau_p}$.
\begin{proposition}\label{prop1}
For a qubit-ergodic channel \(\Lambda^p_{\tau_p}(\rho)\), the expression for ergotropy is given by

\begin{equation}\label{qubit_ergotropy}
\begin{array}{ll}
    W_e\left[\Lambda^p_{\tau_p}(\rho)\right] =\\
    \\
    \frac{e}{2}\left[\sqrt{p^2(r^2 - z^2) + (p(z + z_{\tau}) - z_{\tau})^2} - pz - (1 - p)z_{\tau_p}\right],
\end{array}
\end{equation}
where \(\vec{r} = (x, y, z)\) is the Bloch vector of the state \(\rho\), \(z_{\tau_p}\) is the z-component of the passive state \(\tau_p\), and the system is governed by the Hamiltonian \(H = e|1\rangle\langle1|\) with $e >0$.
\end{proposition}

\begin{proof}
For an arbitrary qubit state \(\rho\), ergotropy is expressed as:\begin{align}
    W_e(\rho) &= E(\rho) - E(\rho_p) \nonumber \\
    &= \left(\frac{1 - z}{2}\right)e - \left(\frac{1 - r}{2}\right)e,
\end{align}
where \(r = \sqrt{x^2 + y^2 + z^2}\) denotes the Bloch vector length of the state \(\rho\).

For the ergodic channel \(\Lambda^p_{\tau_p}(\rho) = \nu = p\rho + (1 - p)\tau_p\), the ergotropy of the evolved state is denoted by:
\begin{align}
    W_e(\nu) &= E(\nu) - E(\nu_p) \nonumber \\
    &= \left(\frac{1 - z'}{2}\right)e - \left(\frac{1 - r'}{2}\right)e,
\end{align}
where \(\vec{r'} = (x', y', z') = (px, py, pz + (1 - p)z_{\tau_p})\) is the evolved Bloch vector for the state \(\nu\). Substituting these values into the ergotropy expression, we obtain Eq. \ref{qubit_ergotropy}.
\end{proof}

\subsection{Thermodynamic Memory Effect in Quantum Batteries}


Ergotropy quantifies the maximum extractable useful energy from an isolated quantum battery via unitary operations, positioning it as a promising resource for future quantum technologies. In Theorem \ref{theorem2}, we demonstrate that during open-system evolution, particularly under Markovian ergodic dynamics, the quantum battery undergoes a gradual self-discharge of ergotropy (useful stored energy) into the environment. This process ultimately leads the system to a passive, fixed-point state \(\tau_p\), rendering the battery's stored energy unusable (see Fig. \ref{fig:ergotropy}a). However, under non-Markovian dynamics, the ergotropy exhibits fluctuations during its evolution. To illustrate this explicitly, we consider a passive state in the qutrit dimension that serves as a fixed point of the EPCPR channel (see Fig. \ref{fig:ergotropy}b). This state is chosen such that its probability distribution deviates from the Gibbsian distribution, \( \frac{\exp(-\beta H)}{Z} \). The origin of these fluctuations is linked to the lack of CP divisibility in the channel dynamics. These fluctuations reveal a memory-assisted energy retention effect, where information flows back from the environment to the system—one of the defining characteristics of non-Markovianity. Notably, ergotropy, a vital resource for quantum batteries, captures this information backflow, demonstrating how non-Markovian dynamics can influence resource behaviors in quantum systems.

This framework also naturally extends to anti-ergotropy, another quantity of interest for quantum batteries. Anti-ergotropy measures the maximum amount of energy that can be added to an isolated quantum battery through unitary operations. For a quantum state \(\zeta\), anti-ergotropy is defined as:  
\[
A_e(\zeta) = \max_U \big[E(U\zeta U^\dagger) - E(\zeta)\big],
\]
where the optimal unitary transfers \(\zeta\) to its corresponding active state. Active states share similarities with passive states—they are diagonal in the energy basis—but exhibit higher probabilities for higher energy levels. For the class of EPCPR channels, if the fixed point is replaced by active states, the channel transforms into an EPCAR (energy preserving channel and active resets) channel. Straightforward calculations reveal that anti-ergotropy exhibits monotonicity under completely positive (CP) divisible dynamics. However, for non-Markovian processes, anti-ergotropy is expected to display memory-induced energy retention during the dynamics, akin to the behavior observed for ergotropy. These findings provide deeper insights into how non-Markovianity impacts in charging quantum batteries, emphasizing the role of memory effects in enhancing resource retention during open-system evolution.

\subsection{Ergotropy as a Measure of Non-markovianity}
Based on the above discussion, it is now necessary to define a measure of non-Markovianity by exploiting ergotropy $W_e(\rho)$. In order to do that, we construct the following quantity 
\begin{equation}\label{nm-ergo}
    \mathcal{B}_W(t) = \frac{d}{dt}W_e^M(t),~~\left[\forall t ~with~\frac{d}{dt}W_e^M(t) > 0\right],
\end{equation}
where $W_e^M(t)=\max_\rho W_e(\zeta)$ is the maximum ergotropy over all possible input states. The maximisation is done to make the measure a state independent quantity. Here $\mathcal{B}_W(t)$ is an information backflow based non-Markovianity measure. To get a normalised time integrated measure, we define 
\begin{equation}\label{nm-ergo1}
\mathcal{N}_W = \frac{N_W}{1+N_W},    
\end{equation}
where 
\[
N_W = \int_{\mathcal{B}_W(t) > 0}\mathcal{B}_W(t)dt.
\]
In the following, we explicitly calculate this measure for qubit passivity EPCPR operation $\Lambda^{p_t}_{\tau}(\rho)$. For such operation, the general expression of ergotropy with Hamiltonian $e\ket{1}\bra{1}$, is given in Proposition \ref{prop1} and exploiting which, we calculate 
\[
\begin{array}{ll}
\mathcal{B}_W(t) = \frac{e}{2}\frac{d}{dt}\left(\frac{\left(2 p_t^2-1\right)z_{\tau}+\sqrt{\frac{p_t+1}{1-p_t}+z_{\tau}^2}}{p_t+1}\right)\\
\\
= \frac{e}{2}\left((1+4p_t+2p_t^2)z_{\tau}+\frac{1}{\sqrt{\frac{p_t+1}{1-p_t}+z_{\tau}^2}}\left(\frac{p_t(1+p_t)}{(1-p_t)^2}-z_{\tau}^2\right)\right)\frac{\dot{p}_t}{(1+p_t)^2}
\end{array}
\]
From the expression, it can be seen that the term in the parenthesis is always positive irrespective of the value of $p_t$ (since $z_\tau \geq 0$ for the given case), and hence $\mathcal{B}_W(t)$ is positive, iff $\dot{p}_t > 0$, i.e. the dynamics is non-Markovian. This observation coincides with our claim presented in Corollary \ref{cor1}. It is also important to mention that our analysis and quantification of non-Markovianity from the perspective of passivity preserving channels are restricted for such channels falling within the class of ergodic channels represented by equation \eqref{ergoN1}. Though the class of channels in question represents a very large set, there are passivity preserving channels not represented by equation \eqref{ergoN1} that are outside the purview of current analysis. This fact encourages further study on the specific topic in question.

\section{conclusion} 

In this work, we have characterized a class of non-unital quantum channels with a singular fixed point, termed ergodic channels, and constructed Lindblad-type master equations to study their non-Markovian dynamics. The structure of these channels enables a comprehensive understanding of both divisibility breaking and information backflow in arbitrary-dimensional systems. Notably, we find that the infinitesimal divisibility of ergodic qubit channels is independent of the channel’s fixed point, a result that is consistent with the RHP measure of non-Markovianity, as infinitesimally divisible channels are solutions of time-dependent GKLS equations.  

For ergodic channels, standard distance measures such as trace distance and relative entropy distance exhibit monotonicity in Markovian dynamics, as expected. However, the physical implications become more pronounced when these channels are linked to thermodynamic systems. For example, thermal ergodic channels have relative entropy distance measures that correspond to free energy differences, directly connecting information-theoretic and thermodynamic quantities. In non-Markovian dynamics, memory retention effects emerge, revealing how work can be recovered as information flows back into the system.

For more general ergodic channels, particularly passivity-preserving channels, since temperature is not well-defined for non-equilibrium states, relative entropy no longer directly corresponds to extractable work as it does when a thermal state serves as the fixed point. Consequently, its conventional thermodynamic interpretation is lost. However, in this scenario, an alternative thermodynamic meaning can be established in terms of ergotropy.Ergotropy, an important figure for quantum batteries, quantifies the maximum extractable useful energy. Under non-Markovian dynamics, ergotropy exhibits memory retention, reflecting the environment's ability to assist in maintaining energy within the system. However, not all passivity-preserving channels exhibit monotonicity in ergotropy, and it remains an open question to determine the full class of such channels for which ergotropy behaves as a monotone.  

In the context of quantum batteries, ergotropy and anti-ergotropy quantify the maximum energy that can be extracted from or supplied to a quantum battery. These quantities are fundamental resources for the operation of quantum batteries. Under Markovian dynamics, these resources are ultimately dissipated to the environment, limiting their utility. However, in non-Markovian dynamics, memory effects allow for the temporary recovery of these resources, offering a potential advantage in optimizing the operation of quantum batteries. For instance, ergotropy may peak at specific points during non-Markovian processes, and by identifying and leveraging these moments, one can extract work from the quantum battery before it is lost to the environment. Similarly, these effects are beneficial in the context of charging processes (anti-ergotropy), which could enhance energy transfer efficiency.  
Nevertheless, understanding the dynamics of non-Markovian processes remains a complex challenge, and fully harnessing these memory effects for practical applications requires a deeper understanding of their underlying dynamics. While the potential for improved energy management in quantum batteries is clear, the realization of these advantages is still limited by the difficulties in precisely characterizing non-Markovian dynamics. Nonetheless, the study of ergotropy in ergodic channels opens up new avenues for exploring the interaction between environmental memory and quantum battery.

\section{Acknowledgements}
We thank the anonymous reviewers for their careful review and helpful comments, which have substantially improved this article. MA acknowledges the funding supported by the European Union  (Project QURES- GA No. 101153001).  Views and opinions expressed are however those of the author(s) only and do not necessarily reflect those of the European Union or European Research Executive Agency. Neither the European Union nor the granting authority can be held responsible for them. HB thanks IIIT Hyderabad for its warm hospitality during his visit  in the initial phase of this work. SB acknowledges the support from the Ministry of Electronics and Information Technology, Government of India, under \textbf{Grant No.} 4(3)/2024-ITEA.

\appendix

\section{Construction of ergodic map in arbitrary dimension}\label{ergoMap-con}

The qubit depolarizing channel can be generalized to arbitrary dimensional channels by generalising the Pauli matrices for higher dimension. Here it is important to note that Pauli matrices are both hermitian and unitary. Therefore, the higher dimensional generalization will retain one of the two properties, but not both. In the following, we present both protocols.

This qubit depolarisng map can be generalized to d-dimension, by replacing the Pauli matrices with Weyl operators \citep{nm25} 
\begin{equation}\label{Weyl}
    \begin{array}{ll}
      \Lambda_w(\rho_{d0})=\sum_{k,l=0}^{d-1}q_{kl}W_{kl}\rho_{d0}W_{kl}^{\dagger} \\
      \\
      \mbox{with}~~W_{kl}=\sum_{m=0}^{d-1}\omega^{mk}\ket{m}\bra{m+l},~~\omega=e^{2\pi i/d} \\
      \\
      \mbox{satisfying}~~W_{kl}W_{rs}=\omega^{ks}W_{k+r,l+s},~~W_{kl}^\dagger = \omega^{kl}W_{-k,-l},
    \end{array}
\end{equation}
with $\rho_{d0}$ being the initial state in $d$-dimension and $q_{kl}$ represents some arbitrary probability. For two dimensions, Weyl channel simplifies to Pauli channel. However, for $d>2$, Weyl channels lose the Hermiticity property of Pauli channel, i.e. $Tr[A\Lambda_P(B)]=Tr[\Lambda_P(A)B]$.
To prevent this, generalizations to higher dimensions require operators $U_\alpha$ constructed from the mutual unbiased bases, (MUB) \citep{nm25}, of the Hilbert space, $\{|\psi_0^{(\alpha)}\rangle,...,|\psi_{d-1}^{(\alpha)}\rangle\}$, of the form
\begin{align*}
    U_\alpha = \sum_{l = 0}^{d-1} \omega^l P_l^{(\alpha)}
\end{align*}
where $P_l^{(\alpha)} = |\psi_l^{(\alpha)}\rangle\langle\psi_l^{(\alpha)}|$ are pure orthogonal projections.
The generalized Pauli channel therefore can be written as
\[\Lambda_{GP}(\rho_{d0}) = p_0 (t) \mathcal{I}(\rho_{d0}) + \frac{1}{d-1} \sum_{\alpha = 1}^{M} p_\alpha(t) U_{\alpha}(\rho_{d0})U_\alpha^\dagger, \]
with $M$ being the maximum number of MUBs in $d$ dimension. 

\section{Construction of the master equation}\label{sec:appendixLB}

Let us consider an invertible dynamical map, of the following form 
\begin{equation}\label{dynamicalmap}
    \rho (t) = \Omega [\rho (0)].
\end{equation}
It is therefore ensured that it corresponds to a Lindblad type master equation 
\begin{equation}\label{eqom}
    \Dot{\rho} (t) = \Tilde{\mathcal{L}} [\rho (t)]
\end{equation}
where $ \Tilde{\mathcal{L}}[.]$ is the generator of the dynamics. Following the method in Ref.\citep{master1,master2}, we now derive the said master equation. 
Let us define $\lbrace \mathcal{G}_i\rbrace$ to be the Hermitian orthonormal basis, with the properties $\mathcal{G}_0 = \mathbb{I}/\sqrt{2}$, $\mathcal{G}_i^{\dagger} = \mathcal{G}_i$, $\mathcal{G}_i$ are traceless with singular exception $\mathcal{G}_0$ and $\text{Tr}[\mathcal{G}_i\mathcal{G}_j] = \delta_{ij}$. The map in Eq.~\eqref{dynamicalmap} can be rewritten as
\begin{equation*}
    \Omega [\rho (0)] = \sum_{m,n} \text{Tr}[\mathcal{G}_m \Omega [\mathcal{G}_n]] \text{Tr}[\mathcal{G}_n \rho (0)] \mathcal{G}_m = [F(t)r(0)] \mathcal{G}^T
\end{equation*}
with $F_{mn}=\text{Tr}[\mathcal{G}_m \Omega [\mathcal{G}_n]]$ and $r_n(s) = \text{Tr}[\mathcal{G}_n \rho(s)]$. The time-derivative of the previous equation gives us
\begin{equation*}
    \dot{\rho}(t) = [\dot{F}(t)r(0)] \mathcal{G}^T.
\end{equation*}
Further consider a matrix $L$ with elements $L_{mn} = \text{Tr}[\mathcal{G}_m \Tilde{\mathcal{L}} [\mathcal{G}_n]]$. Therefore Eq.~\eqref{eqom} can be written as
\begin{equation}\label{Eq13}
    \dot{\rho}(t) = \sum_{m,n} \text{Tr}[\mathcal{G}_m]\Tilde{\mathcal{L}} [\mathcal{G}_n] \text{Tr}[\mathcal{G}_n \rho (t)] \mathcal{G}_m = [L(t)r(t)] \mathcal{G}^T.
\end{equation}
By comparison, we get
\begin{equation*}
    \dot{F}(t) = L(t) F(t) \implies L(t)= \dot{F}(t) F(t)^{-1}.
\end{equation*}
One may note that $L(t)$ can be obtained if $F(t)^{-1}$ exists and $F(0) = \mathbb{I}$. Once we get the $L$ matrix, the master equation of the form in equation \eqref{Eq13} follows directly, which can be written in the Lindblad form \citep{master1,master2}.

\section{Divisibility of the qubit Ergodic map}\label{sec:appendixDiv}
While trying to apply a channel $\Lambda$ on a system $\rho$, one might be constrained to evolve the system under certain dynamical classes. For instance, if one is limited to using Markovian evolutions, one might ask the question: what class of channels can be simulated by such a dynamics? In this section, we ask the following question: Can a given quantum channel be simulated by a Markovian process? As we have seen in the previous subsection, the concept of Markovainity of a dynamical quantum channel, especially as quantified through the RHP measure, is intricately related to the divisibility of the channel. Therefore, it behooves us to investigate the divisibility properties of a quantum channel, in effect questioning when a given channel can be constructed using other channels. For instance, it should not be surprising that a channel that can be written as a product of infinitely many channels, that are close to identity, will always be a solution of a GKLS equation, for such channels can be constructed by applying such infinitesimal channels one after the other. This property of \textit{infinitesimal divisibility} is therefore a sufficient condition for a channel to be a solution of a time-dependent GKLS equation. It turns out that for qubit channels, infinitesimal divisibility is a necessary and sufficient condition\ \cite{wolf2008dividing}.
\\
In this section, we will explore the divisibility properties of the channels of the form
\begin{equation}\label{eq:ergchnl}
    \Lambda_\tau(\rho)= p\rho+(1-p)\tau,
\end{equation}
where, $\tau$  is a fixed density matrix and $p\in [0,1]$ can be interpreted as a probability function. A qubit state is uniquely characterized by its Bloch vector: $\rho=\dfrac{1}{2}(\mathbb{I}+\sum_i a_i\sigma_i)$  which can be represented by a vector $|\rho\rangle\rangle= (1,a_1,a_2,a_3)^T$. The action of any qubit channel, say $\Lambda$, on $\rho$ can be represented as a linear transformation of its Bloch vector. Let $|\Lambda(\rho)\rangle\rangle=(1,a_1',a_2',a_3')^T$, then the action of the channel can be written as a matrix equation as follows
\begin{equation}
    \begin{aligned}
        T_\Lambda|\rho\rangle\rangle = |\Lambda(\rho)\rangle\rangle
    \end{aligned}
\end{equation}
where the $T$-matrix \cite{ruskai2002analysis} of the channel has the form
\begin{equation}
    T_\Lambda=\begin{pmatrix}
        1 & \vec{0}^T\\
        \vec{v} & R
    \end{pmatrix}.
\end{equation}
In the Bloch sphere picture, $R$ performs a three-dimensional rotation of the Bloch vector $\vec{a}=\{a_1,a_2,a_3\}$, while the vector $\vec{v}$ is the displacement of the center of the Bloch sphere. At least for a full Kraus rank channel, it is known that this representation is related to the Choi state of the map by \cite{davalos2019divisibility} 
\begin{equation}
    C_\Lambda=\sum_{i,j=0}^{3} {T_{\Lambda}}_{ij}\sigma_{i}\otimes\sigma_{j}
\end{equation}
For the qubit channel given in\ \eqref{eq:ergchnl}, let $\tau=\dfrac{1}{2}(\mathbb{I}+\sum_i b_i\sigma_i)$, then the T-matrix for this channel is given by the following
\begin{equation}\label{Tmatrix}
    T_\Lambda=\begin{pmatrix}
    1& 0& 0& 0\\
   (1-p) b_1& p& 0& 0\\
    (1-p)b_2& 0& p& 0\\
    (1-p) b_3& 0& 0& p
\end{pmatrix} .
\end{equation}

\textit{P-divisibility:} A channel is called $P-$divisible if it can be written as a product of positive maps. A dynamical map which is Markovian under the BLP measure is, therefore, always $P$-divisible since trace distance is a monotone under a positive map. For a positive map, the determinant of its $T$-matrix is positive\ \cite{ruskai2002analysis}. Since the concatenation of channels is simply a matrix multiplication of their T-matrices,  one can conclude that the positivity of the determinant of $T_\Lambda$ is a necessary condition for the $P-$divisibility of the channel $\Lambda_\tau$. Following \citep{wolf2008dividing}, it is also a sufficient condition for the $P$-divisibility of the qubit channels. Therefore, from the form of the $T$-matrix of the $\Lambda_\tau$ given in equation  \eqref{Tmatrix},  it is clear that such channels are always P-divisible, independent of the state $\tau$, since the determinant of the $T_\Lambda$-matrix is simply given by $p^3$. 
\\\\
\textit{Infinitesimal CP-divisibility}: If a given quantum channel $\Lambda$ can be written as a composition of quantum channels that are arbitrarily close to the Identity channel, we call the channel $\Lambda$ an infinitesimally CP-divisible channel. More precisely, if $\forall\ \epsilon >0$, there exists a set of channels $\Lambda_i$ s.t. $\Lambda=\prod_i \Lambda_i$ where $||\Lambda_i-I||_1 \le \epsilon$, then $\Lambda$ is said to be infinitesimally divisible.
\\
The CP-divisible dynamical channels $\Lambda_t$ s.t. $\Lambda_t=\Lambda_{t,s}\circ\Lambda_{s,0}$ are known to have a generator  $\mathcal{L}_t$ that satisfies the Lindbladian equation. Therefore, the map can be written as 
$$
\Lambda_{t,0}=\hat{T}e^{\int_0^t \mathcal{L}_l dl}.
$$
Since any map that is close to identity is approximately equal to a map with a Lindbladian generator, $\Lambda_{\delta}=I+\mathcal{L}_\delta\sim e^{\mathcal{L}_\delta}$, one can show (see\ \cite{wolf2008dividing}) that all infinitesimally divisible maps can be arbitrarily approximated in terms of the CP-divisible maps: $\mathcal{E}\sim \prod_\epsilon e^{\mathcal{L}_\delta}$. In other words, infinitesimally divisible channels are solutions of the time-dependent GKLS equation.
\\ \\
For any qubit channel, it is infinitesimally CP-divisible if and only if, the singular values $\{s_1,s_2,s_3\}$ of the Lorentz normal form of its $T$- matrix satisfy the following relation \cite{wolf2008dividing}:
$$s_{min}^2\ge s_1s_2s_3 > 0.$$
In the appendix \ref{app:lorentznormal}, we calculate the Lorentz normal form and its singular values for the ergodic channel given in equation\ \eqref{eq:ergchnl}.
\begin{figure}[h!]
    \centering
    \includegraphics[width=1\linewidth]{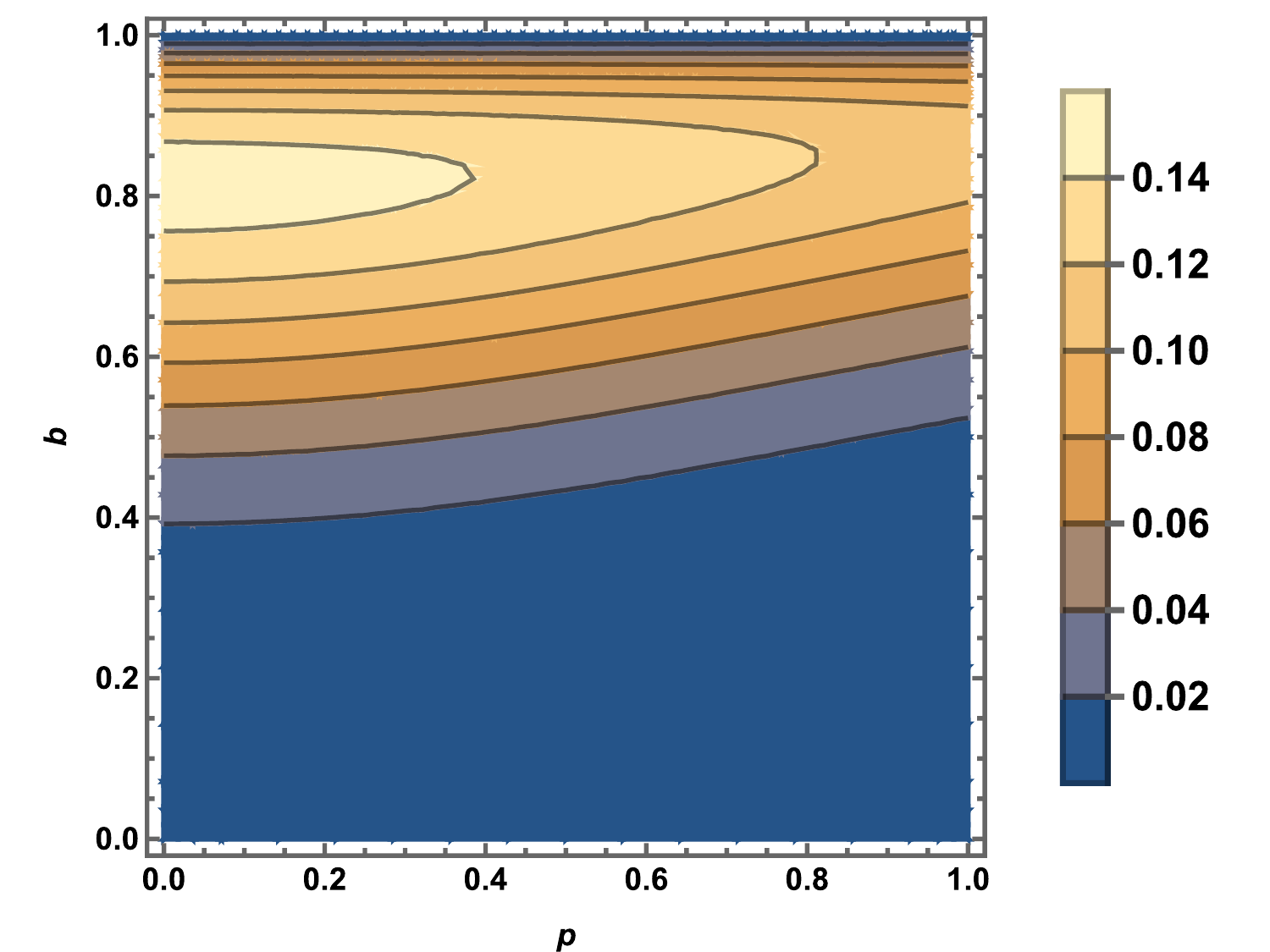}
    \caption{Contour plot of the difference $s_{min}^2-s_1s_2s_3$ as a function of $b$ (the length of the Bloch vector of the state $\tau$) and the probability $p$. We see that the infinitesimal divisibility condition is obeyed for all such channels.}
    \label{contour1}
\end{figure}
In figure \ref{contour1} we plot the difference $s_{min}^2-s_1s_2s_3$ as a function of the length of the Bloch vector $b$ and the probability value $p$. We see that the ergodic channel is always an infinitesimally CP-divisible channel, independent of the state $\tau$ or the probability function $p$, and therefore, it is always a solution of a CP-divisible dynamics. In other words, a qubit ergodic channel can be achieved through a Markovian process. This also tells us that the non-Markovianity of the dynamical qubit ergodic channel must depend only on the time-dependent function $p(t)$. The independence of the infinitesimal divisibility on the state $\tau$ is consistent with the result in the previous sub-section that the RHP measure of non-Markovianity of the ergodic channel is independent of the fixed state.

\section{Calculating the Lorentz normal form}\label{app:lorentznormal}
For a general CPTP qubit channel $\Lambda$, with the Pauli basis representation $T_\Lambda=\begin{pmatrix}
        1& \vec{0}^T\\
        \vec{v}& R
    \end{pmatrix}$ , $\exists$ rank-1 completely positive maps $L_1$ and $L_2$ such that $\Sigma=L_1T_\Lambda L_2$ is proportional to $\begin{pmatrix}
        1& \vec{0}^T\\
        \vec{v'}& \Delta
    \end{pmatrix}$ and has one of the three forms (called the Lorentz normal form) \cite{verstraete2001local}:
    \begin{enumerate}
        \item $v'=0$ and $\Delta$ is diagonal.
        \item $v'=(0,0,2/3)$ and $\Delta=\text{diag}(x/\sqrt{3},x/\sqrt{3},1/3)$. 
        \item $v'=(0,0,1)$ and $\Delta=0$. This is a Kraus rank 2 channel and corresponds to a pin map.
    \end{enumerate}
We will restrict to the case where $T_\Lambda$ is a full rank matrix defined in equation \eqref{Tmatrix}. Define a Hermitian operator $C=T_\Lambda T_\Lambda^\dagger$. Diagonalizing this matrix we get $P^{-1}CP=D=\text{diag}(e_1,e_2,e_3,e_4)$, where the eigenvalues of $C$ are given by:
\begin{equation}
\begin{array}{ll}
\begin{pmatrix}
    e_1\\ e_2 \\ e_3 \\ e_4
\end{pmatrix}=\\
\\
\begin{pmatrix}
1+p ^2 +(1-p) ^2 b^2+\sqrt{(1-p) ^2 \left(b^2+1\right) \left((1-p) ^2 b^2+(p +1)^2\right)}/2\\
    p^2\\
    p^2\\
  1+p ^2 +(1-p) ^2 b^2-\sqrt{(1-p) ^2 \left(b^2+1\right) \left((1-p) ^2 b^2+(p +1)^2\right)}/2
\end{pmatrix}
\end{array}
\end{equation}
Here $b$ is the length of the Bloch vector of the state $\tau$ and $p$ is a probability function. We note that for all values of $b$ and $p$, the eigenvalues are ordered as $e_1\ge e_2= e_3 \ge e_4$.
Similarly, define $C'=T_\Lambda^\dagger T_\Lambda$ such that $P'^{-1}C'P'=D$. Note that $P$ and $P'$ are Lorentz transformations since $P^T P=I$ and similarly for $P'$. One can see that $PT_\Lambda P'=\sqrt{D}$, and therefore, defining $L_1=P/\sqrt{e_1}$ and $L_2=P'$ we have $\Sigma=\sqrt{D}/\sqrt{e_1}$, which has the Lorentz normal form of type 1. We have  $\Delta=\text{diag}(s_1,s_2,s_3)$ where
$$
\begin{pmatrix}
    s_1\\s_2\\
    s_3
\end{pmatrix}=
\begin{pmatrix}
    \sqrt{e_2/e_1}\\
    \sqrt{e_3/e_1}\\
   \sqrt{e_4/e_1}
\end{pmatrix}.
$$
Here we have $s_1=s_2\ge s_3$ for all values of $p$ and $b$. 
\bibliographystyle{apsrev4-1}
\bibliography{sorsamp}
\end{document}